\documentclass[12pt]{article}
\RequirePackage[hidelinks]{hyperref}
\usepackage[top=1in, bottom=1in, left=1in, right=1in]{geometry} 
\usepackage{amsthm,soul,cancel,mathrsfs,euscript,amsfonts,natbib,amsmath,xcolor,color,amssymb,amsmath,graphicx,float,booktabs,subcaption,amsthm,bm,bbm,multirow,threeparttable,textcomp,adjustbox,setspace,url,enumerate,graphicx,psfrag,epsf,algorithm}
\theoremstyle{plain}
\newtheorem{lemma}{Lemma}
\newtheorem{theorem}{Theorem}
\newtheorem{corollary}{Corollary}
\newtheorem{remark}{Remark}

\def\bw{{\bf w}}
\def\tr{\mbox{tr}}
\def\bepsilon{\boldsymbol{\epsilon}}
\def\betatilde{{\widetilde\beta}}
\def\bPitilde{\widetilde\Pi}
\def\Bsc{{\cal B}}
\def\Csc{{\cal C}}
\def\Dsc{{\cal D}}
\def\jtilde{{\widetilde j}}
\def\half{\frac{1}{2}}
\def\Abb{\mathbb{A}}
\def\Ibb{\mathbb{I}}
\def\bA{{\bf A}}
\def\bU{{\bf U}}

\def\bX{{\bf X}}
\def\bY{{\bf Y}}
\def\bZ{{\bf Z}}
\def\Usc{{\cal U}}
\def\Wbbhat{\widehat{\mathbb{W}}}
\def\bPi{\boldsymbol{\Pi}}
\def\bpi{\boldsymbol{\pi}}
\def\bmu{\boldsymbol{\mu}}
\def\argmindum{\mathop{\mbox{argmin}}}
\def\argmin#1{\argmindum_{#1}}

\def\argmax#1{\argmaxdum_{#1}}

\def\diag{\mbox{diag}}
\def\transpose{{\sf \scriptscriptstyle{T}}}
\def\trans{^{\transpose}}

\def\onetomany {one-to-many }
\def\onetoone {one-to-one }
\def\supone{^{\scriptscriptstyle \sf [1]}}
\def\suptwo{^{\scriptscriptstyle \sf [2]}}

\def\Vbb{\mathbb{V}}

\def\cos{\mbox{cos}}
\def\sumin{\sum_{i=1}^n}
\def\ellhat{\widehat{\ell}}
\def\Wbbhat{\widehat{\Wbb}}
\def\argmax{\mbox{argmax}}

\def\Initialpi{\widetilde{\bPi}}

\def\Xbb{\mathbb{X}}
\def\Ybb{\mathbb{Y}}
\def\Ubb{\mathbb{U}}

\def\Vbb{\mathbb{V }}
\def\Wbb{\mathbb{W}}
\def\addstar{} 
\def\addstartop{^\top} 
\def\bPihat{\widehat{\bPi}}
\def\Ebb{\mathbb{E}}
\def\hypersphere{ \EuScript{S}^{p-1}}
\def\mm{\mathcal{S}}
\def\boundgamma{\rho}
\def\ii{_{[\mm(\bPihat\suptwo),:]}}
\def\Gsc{[n]}
\def\nmis{n_{\sf\scriptscriptstyle mis}}

\begin{document}
	\def\spacingset#1{\renewcommand{\baselinestretch}%
		{#1}\small\normalsize} \spacingset{1}
	 \title{\Large\bf Spherical Regression under Mismatch Corruption with Application to Automated Knowledge Translation}
	\author{Xu Shi$^1$, Xiaoou Li$^2$, and Tianxi Cai$^3$ \\
		$^1$Department of Biostatistics, University of Michigan
		\\$^2$Department of Statistics, University of Minnesota
		\\$^3$Department of Biostatistics, Harvard University}
	\date{}
	\maketitle
	
	\bigskip
	\begin{abstract}
		Motivated by a series of applications in data integration, language translation, bioinformatics, and computer vision, we consider spherical regression with two sets of unit-length vectors when the data are corrupted by a small fraction of mismatch in the response-predictor pairs. We propose a three-step algorithm in which we initialize the parameters by solving an orthogonal Procrustes problem to estimate a translation matrix $\Wbb$ ignoring the mismatch. We then estimate a mapping matrix aiming to correct the mismatch using hard-thresholding to induce sparsity, while incorporating potential group information. We eventually obtain a refined estimate for $\Wbb$ by removing the estimated mismatched pairs.
		We derive the error bound for the initial estimate of $\Wbb$ in both fixed and high-dimensional setting. We demonstrate that the refined estimate of $\Wbb$ achieves an error rate that is as good as if no mismatch is present. We show that our mapping recovery method not only correctly distinguishes one-to-one and one-to-many correspondences, but also consistently identifies the matched pairs and estimates the weight vector for combined correspondence. 
		We examine the finite sample performance of the proposed method via extensive simulation studies, and with application to the unsupervised translation of medical codes using electronic health records data. 
	\end{abstract}
	
	\vfill
	\newpage
	\spacingset{1.5} 
	\section{Introduction}
	
	Classical multivariate regression analysis studies the relationship between a response random vector and a predictor random vector, under the assumptions that the response-predictor pairs are correctly linked, and that the data lie in an unrestricted Euclidean space. However, modern large-scale datasets are frequently integrated from multiple heterogeneous data sources. Observations from different datasets are often imperfectly matched due to linkage error. In addition, in many real-world settings ranging from gene expression analysis to language processing, the response and predictor vectors represent directional data, which lie on the surface of a hypersphere \citep{gotsman2003fundamentals,xing2015normalized}. Motivated by the applications in automated translation of medical code, we propose in this paper novel multivariate regression procedures for spherical data in the presence of mismatch. We first detail the motivating examples and then discuss the statistical contributions of the paper.
	
	\subsection{Automated translation of medical codes \label{ICD-9motivation}}
	A motivating example is the translation of medical codes routinely documented in electronic health records (EHR). 
		An EHR is a digital version of a patient's medical records, which contain rich clinical information including medical history, diagnoses, medications, treatments, immunization, allergies, radiology images and laboratory test results.
	The Centers for Medicare and Medicaid Services (CMS) recently renamed the EHR Incentive Program from ``meaningful use" to ``promoting interoperability", aiming to improve the integration and sharing of health information among providers, clinicians, and patients. A key challenge is the lack of semantic interoperability because the ``languages" used in different EHR systems and across time may be inconsistent.
	For example, the International Classification of Diseases (ICD) codes describe medical diagnoses and procedures for billing purposes. Data on ICD codes are used extensively for biomedical research \citep[e.g]{yu2015toward,chen2013applying,parle2001prediction}. However, due to the coding incentives and the heterogeneity in healthcare systems, different providers may use alternative codes to record the same diagnosis or procedure, limiting the transportability of phenotyping algorithms and prediction models across systems. Translation of ICD codes between different healthcare systems can potentially overcome such challenges.
	
	Another example of code translation arises from the updating of ICD coding systems. All U.S. healthcare systems are federally mandated in 2015 to replace the 9th edition of ICD codes (ICD-9) with the 10th edition (ICD-10) for all claims of service, with a potential to convert to the 11th edition in 2022 \citep{icd11}.  Mappings between ICD-9 and ICD-10 codes are essential for linking and analyzing EHR data before and after the transition. Available manual annotations such as the General Equivalence Mappings (GEM) are intrinsically ambiguous due to the increase in specificity and number 
	of ICD-10 codes \citep{krive2015complexity}. In particular, a significant portion of the GEM mappings are \onetomany  mapping, and many are approximate matches. For example, the ICD-9 code 995.29 ``{\em unspecified adverse effect of other drug, medicinal and biological substance}" is mapped to over a hundred ICD-10 codes. The presence of  \onetomany  mapping and the inherent differences between the two coding systems pose substantial challenges to the translation of ICD-9 codes to ICD-10 codes.

	Manual translation of medical codes is not only immensely laborious but also error prone, signifying the need for data-driven translation methods. In this paper, we turn the problem of code translation into a statistical problem of mapping two sets of unit-length vectors, $\Ybb=[\bY_1, ..., \bY_n]\trans$ from one system and $\Xbb=[\bX_1, ..., \bX_n]\trans$ from another system, where $\bY_i$ and $\bX_i$ respectively represent semantic embedding vector (SEV) for the $i^{th}$ medical code in the two systems. The SEVs are generated from the \texttt{word2vec} word embedding algorithm, which essentially learns the interpretation of the medical codes from their co-occurrence patterns in the EHR data which reflect clinical practice \citep{mikolov2013distributed}. See Section~\ref{supp:W2Valgorithm} of the Supplementary Material and \cite{beam2018clinical} for details on the training of SEVs. 
			
			For example, Figure~\ref{grouping} presents select ICD-9 code SEVs from two healthcare systems, the Partners HealthCare System and the Veterans Health Administration. The ICD-9 codes are grouped into clinically meaningful phenotypes according to the ICD-to-phenotype mapping from the phenome-wide association study (PheWAS) catalogue \citep{denny2010phewas}. 
			Each point in Figure~\ref{grouping} represents an ICD-9 code SEV, which is color-coded by the PheWAS group.
			The directions of the SEVs encode the relationship, similarity, and clinical meaning of the codes. Particularly, SEVs of codes with more similar meanings are closer to each other. 
			We thus propose to achieve code translation by inferring a mapping between the two sets of data-driven embeddings, $\Ybb$ and $\Xbb$.
	\begin{figure}[!h]
		\centering
		\begin{subfigure}[t]{0.48\textwidth}
			\centering\includegraphics[width=0.5\textwidth]{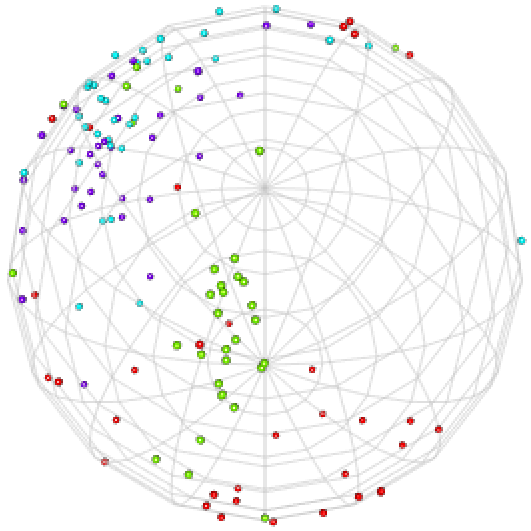}
			\caption{\label{fig:f1}
				Veterans Health Administration (VHA)}
		\end{subfigure}
		\begin{subfigure}[t]{0.48\textwidth}
			\centering\includegraphics[width=0.5\textwidth]{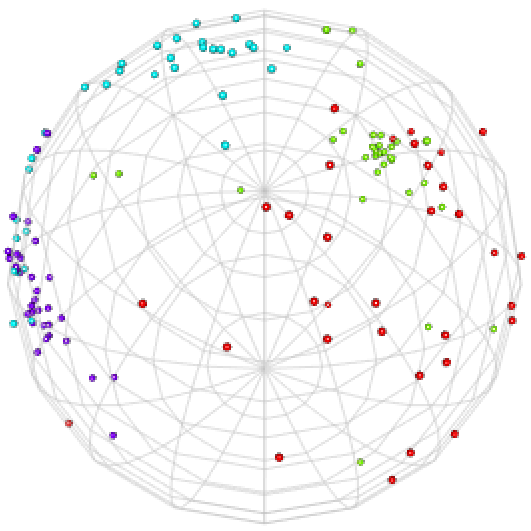}
			\caption{\label{fig:f2}
				Partners HealthCare Systems (PHS)}
		\end{subfigure}
		\caption{\label{grouping} First three principal components of ICD-9 code semantic embedding vectors in four select PheWAS groups from Veterans Health Administration and Partners HealthCare Systems.  Each point represents an ICD-9 code, and is color-coded by the PheWAS group.}
	\end{figure}
	
	In addition to medical code translation, regression with mismatched spherical data has applications in many other scientific problems. Examples include language processing \citep{xing2015normalized,wilson2015controlled}, bioinformatics \citep{sael2010binding,samarov2011local}, pose and correspondence determination in image processing \citep{gold1995new,zhou2014vision}, simultaneous localization and mapping in robotics \citep{kaess2015simultaneous,esteves20173d}, 
	shape matching and retrieval \citep{kazhdan2003rotation,papadakis2007efficient} and computer vision and pattern recognition \citep{marques2009subspace,cohen2018spherical}.

	\subsection{Spherical Regression with Mismatched Data}\label{sec:sphmis}
	
	We propose to create a mapping between the code-SEVs allowing for both \onetoone  and \onetomany  correspondences by developing a spherical regression model with mismatched data. Specifically,  we assume that $\bY_i$ relates to $\Xbb=[\bX_1, ..., \bX_n]\trans$ only through $(\bPi_{i\cdot} \Xbb\Wbb)\trans$, where $\bX_i$ and $\bY_i$ lie on the surface of a $p$-dimensional unit sphere denoted by $\hypersphere$, $\Wbb\in \mathcal{R}^{p\times p}$ is an orthogonal translation matrix satisfying $\Wbb\Wbb\trans = \Ibb_p$ with $\Ibb_p$ an identity matrix, and $\bPi = [\bPi_{1\cdot}\trans, ..., \bPi_{n\cdot}\trans]\trans\in \mathcal{R}^{n\times n}$ is a mapping matrix that corrects the potential mismatch.

	There is a growing literature on the shuffled linear regression problem of $\bY_i = (\bPi_{i\cdot} \Xbb\Wbb)\trans + \bU_i$ when $\bPi$ is a {\em permutation} matrix encoding only \onetoone  correspondence between $\Xbb$ and $\Ybb$ and no orthogonality constraint is imposed on $\Wbb$ \citep[e.g.]{pananjady2017denoising,pananjady2017linear,slawski2017linear,abid2017linear,hsu2017linear,unnikrishnan2018unlabeled}.
	It has been shown that the least squares estimator of $\Wbb$ is generally inconsistent without any additional constraints imposed on $\bPi$ \citep{pananjady2017denoising,pananjady2017linear,slawski2017linear}. When $\bPi$ is sparse in that only a small portion of the responses or predictors is permuted, $\Wbb$ can be consistently estimated \citep{slawski2017linear}. Algorithms for estimation of $\Wbb$ have also been studied \citep{hsu2017linear,abid2017linear,unnikrishnan2018unlabeled}.
	Estimation of the permutation matrix $\bPi$ is challenging both computationally and statistically.
	Specifically, permutation recovery is generally NP-hard unless $p=1$ or $ \bU_i = 0$ \citep{pananjady2017linear,hsu2017linear}. When $p=1$, estimation of $\bPi$ reduces to a sorting problem and thus is computationally tractable. Statistical limit in terms of conditions on the signal-to-noise ratio (SNR) required for the recovery of $\bPi$ has also been studied \citep{pananjady2017linear,slawski2017linear,hsu2017linear}.
	
	While existing literature on regression with mismatched data generally assumes Gaussian data with a random or fixed design matrix, this paper concerns the case where both $\bX_i$ and $\bY_i$ belong to $\hypersphere$. With perfectly matched data in the spherical domain $\hypersphere$, estimation of an orthogonal matrix $\Wbb\in SO(p)=\{A\in \mathcal{R}^{p\times p}: AA\trans =\Ibb_{p}\}$ that transforms the predictors to responses  has been referred to as the spherical regression 
	\citep{chang1986spherical,chang1989spherical,goodall1991procrustes,kim1998deconvolution,rosenthal2014spherical,di2018nonparametric}. Statistical inference beyond the classical setup of fixed dimension $p$ has also been considered recently \citep{paindaveine2017detecting}. However, the current literature is based on the assumption that the response and predictor are correctly linked. 
	
	In this paper, we fill the gaps by developing estimation procedures for $\Wbb$ and $\bPi$ with mismatched spherical data. 
	Instead of imposing  \onetoone  correspondence for $\bPi$, we focus on the setting where $\bPi$ is sparse with a block diagonal structure allowing for both \onetoone  and \onetomany  mappings. Specifically, we assume that the group information is available and mismatch is only expected to occur within a group. 
	In ICD code translation, for example, the phenotype or disease categories can be used as group information. Codes belonging to one disease category (e.g. rheumatoid arthritis) in one healthcare system will never be mapped to codes belonging to a different disease category (e.g. type II diabetes) in another healthcare system. Additional examples of grouping include pathway information for gene expression, brain regions for brain imaging, and semantic groups for clinical concepts in the Unified Medical Language Systems \citep{lindberg1993unified}.
	The group structure may not ease estimation of $\Wbb$ but can greatly reduce the difficulty in recovering $\bPi$. To the best of our knowledge, no existing method consider the recovery of a general mapping matrix leveraging group information. 
	The rest of the paper is organized as follows. We detail our model assumptions and estimation procedures in Section~\ref{method}. 
	In Section~\ref{sec:theory}, we investigate how the degree of mismatch influences the error rates, and we detail theoretical guarantees for our proposed method.
	We evaluate the performance of our proposed method via extensive simulation studies in Section~\ref{simu}. In Section~\ref{application} we apply the proposed method to translate ICD-9 codes between two healthcare systems using SEV data derived from the corresponding EHRs and to translate between ICD-9 and ICD-10 codes using SEV data derived from the same EHR system.
	We close with a discussion in Section~\ref{discussion}.

	\section{Method}\label{method}
	\subsection{Notation}\label{notation}
	We assume that the data consist of $n$ pairs of $p$-dimensional unit-length vectors in $\hypersphere$, i.e., $\Ybb=[Y_{ik}]_{n\times p}=[\bY_{1}, ..., \bY_{n}]\trans $ and $\Xbb=[X_{ik}]_{n\times p}=[\bX_{1}, ..., \bX_{n}]\trans$. The $n$ observations belong to $K$ groups indexed by $\{G_1, ..., G_K\} \subset \Gsc = \{1, ..., n\}$ and mismatch only occurs within group. Let $n_k = |G_k|$ denote the group size with $\sum_{k=1}^K n_k=n$, where for an index set $G$, $|G|$ denotes its cardinality. Without loss of generality, we assume that the observations are ordered by group and thus $\bPi = \diag\{\bPi^1, ..., \bPi^K\}$, where  $\bPi^{k}$ denotes the matrix that encodes the mapping among records within $G_k$. For indexes $i,j\in\Gsc $, let $i\sim j$ denote that $i$ and $j$ belong to the same group, i.e., $i,j\in G_k$ for some $k$. 
	
	For a matrix $\bA$, let $\bA_{i\cdot}$ and $\bA_{\cdot j}$ respectively denote its $i^{th}$ row and $j^{th}$ column, $\sigma_{i}(\bA)$ denote the $i^{th}$ largest singular value of $\bA$, and $\|\bA\|_F$ denote the Frobenius norm of $\bA$. For an index set $G$, let $\bA_{[G,:]}$ denote the rows of $\bA$ corresponding to $G$. Let $\|\cdot\|_2$ denote the $\ell_2$ norm of a vector. Let $\Ibb_n$ denote the $n\times n$ identity matrix, and we omit $n$ when it is self-explanatory. 
	For any mapping matrix $\bpi\in\mathcal{R}^{n\times n}$, let $\mm(\bpi)=\{i \in \Gsc: \bpi_{i\cdot}=\Ibb_{i\cdot} \}$ and $\Dsc(\Ibb, \bpi) = \{i \in \Gsc: \bpi_{i\cdot}\ne \Ibb_{i \cdot} \}$ respectively index the set of matched and mismatched units as determined by $\bpi$, with $\Dsc(\Ibb,\bpi)=\mm(\bpi)^c$ where for any set $\mm$, $\mm^c$ denotes its complement. Accordingly, let $\nmis = |\Dsc(\Ibb, \bPi)|$ denote the number of mismatched pairs in the data. 
	
	\subsection{Model Assumptions}
	\subsubsection{Spherical data and the von-Mises Fisher distribution\label{vMFdefine}}
	
	Unlike the Euclidean space, the $\hypersphere$ sample space features distinctive characteristics both theoretically and practically. 
	The most widely used distribution family for random vectors in $\hypersphere$ is the von-Mises Fisher (vMF) distribution. The $p$-dimensional vMF distribution with parameters $\bmu$ and $\kappa$, denoted by vMF$_{\bmu,\kappa,p}$, has density 
	\begin{equation}\label{vmf}
		f_{\text{vMF}}(\bY|\bmu;~\kappa) = C_{p}(\kappa)\exp(\kappa \bmu\trans  \bY) =  C_{p}(\kappa)\exp\{\kappa \cos(\bmu, \bY)\},
		\end{equation}
	where $\kappa\geq 0$ is a concentration parameter, $\bmu\in \mathcal{R}^p$ is the mean direction with $\|\bmu\|_2=1$, $C_{p}(\kappa)=\kappa^{p/2-1}/\{(2\pi)^{p/2}B_{p/2-1}(\kappa)\}$, and $B_{p/2-1}(\cdot)$ denotes the modified Bessel function of order $p/2-1$. The vMF distribution belongs to the exponential family and thus has many desirable statistical properties. For example, one can show that if $\bZ\sim N(\bmu,\Ibb_p/\kappa)$, then conditional on having unit length, $\bZ\big| \|\bZ\|_2\!=\!1$ follows vMF$_{\bmu,\kappa,p}$ distribution. In addition, for a random vector $\bY\sim \text{vMF}_{\bmu,\kappa,p}$, we have $E[\bY] = \gamma_{\kappa,p}\bmu$ and $E[\|\bY-E[\bY]\|^2_2]=1-\gamma_{\kappa,p}^2$, where 
	$\gamma_{\kappa,p}= B'_{p/2-1}(\kappa)/B_{p/2-1}(\kappa)-(p/2-1)/\kappa$ can be bounded as in the following lemma: 
	\begin{lemma}\label{lemma:eta-bound}
		For $p\geq 4$ and $\kappa>0$, $\max\{0,1-\frac{p-1}{2\kappa}\}< \gamma_{\kappa,p}< 1$.
	\end{lemma}
	The above results are proved in Section~\ref{app-lemma} of the Supplementary Material.
	Intuitively, random vectors following the $\text{vMF}_{\bmu,\kappa,p}$ distribution are symmetrically distributed on $\hypersphere$ concentrating around the mean direction $\bmu$. The expectation is of the same direction as $\bmu$ but lies inside the sphere, i.e., $\gamma_{\kappa,p}< 1$. As the distribution gets more concentrated around $\bmu$,  the expectation gets closer to $\bmu$.
	In addition, the large deviation bounds for sums of i.i.d copies of $\|\bY-\bmu\|_2^2$ derived in Proposition~\ref{lemma:tail-bound-von-miss} of the Supplementary Material may be of independent interest.

	\subsubsection{Unified loss function on the hypersphere\label{norm}}
	The spherical data are also unique in that the loss function defined on the hypersphere unifies a lot of commonly used distance measures. Here we formally introduce our objective function for estimating $\Wbb$ and illustrate such unifying property. 
	To ease exposition, we first consider a simplified scenario with $\bPi=\Ibb_n$ under which we may estimate the translation matrix $\Wbb$ by minimizing the Frobenius norm
	\begin{equation}\label{modelW}
		\Wbbhat = \argmin{\Wbb: \Wbb\Wbb\trans  = \Ibb_p}{\ellhat_0(\Wbb)},~\text{where } \ellhat_0(\Wbb)= \| \Ybb- \Xbb \Wbb\|_F^2.
		\end{equation}
	The role of $\Wbb$ is to align the spaces spanned by columns of $\Xbb$ and $\Ybb$ such that samples in $\Ybb$ and $\Xbb\Wbb$ can be compared in distance.
	The orthogonal parameterization $\Wbb\Wbb\trans =\Ibb$ ensures that the transformed data remain on the sphere, i.e. $\|\Wbb\trans \bX_{i}\|_2=\|\bX_{i}\|_2=1$.

	Because both $\bX_{i}$ and  $\bY_{i}$ have unit length, minimizing the loss function is equivalent to maximizing the cosine similarities between $\bY_{i}$ and its transformed counterpart $\Wbb\trans \bX_{i}$. 
	In addition, the cosine similarity is equal to the inner product when the vectors are of unit lengths. To summarize, we have the following equivalence
	\begin{equation*}
		\argmin{\Wbb: \Wbb\Wbb\trans  = \Ibb_p}{\ellhat_0(\Wbb)} 
		=\underset{\Wbb: \Wbb\Wbb\trans  = \Ibb_p}{\argmax}{\;\sumin \cos(\bY_{i}, \Wbb\trans \bX_{i})}
		=\underset{\Wbb: \Wbb\Wbb\trans  = \Ibb_p}{\argmax}{\;\sumin \bY_{i}\trans \cdot (\Wbb\trans \bX_{i})}.
		\end{equation*}
	The loss function $\ellhat_0(\Wbb)$ also corresponds to the log-likelihood function under the vMF distribution. Specifically, $\ellhat_0(\Wbb)$ corresponds to the log-likelihood function under the model
	\begin{equation}\label{nopi}
		f_{\text{vMF}}(\bY_i|\Xbb;~\kappa) = C_{p}(\kappa)\exp(\kappa \bmu_i\trans  \bY_i) \quad \mbox{with}\quad \bmu_i = \Wbb\trans\bX_i = \Wbb\trans(\Ibb_{i\cdot}\Xbb)\trans \ \mbox{and} \ \Wbb\Wbb\trans = \Ibb_p
		\end{equation}
	with $\bY_i|\Xbb,i\in\Gsc$ independent.
	We thus target an objective on the hypersphere $\hypersphere$ unifying the Frobenius norm, the cosine similarity, the inner product, and the likelihood function of the von Mises-Fisher distribution.
	
	\subsubsection{Model Assumptions under Mismatch with Group Structure \label{sparsepi}}
	Building upon the above objective, we consider the general scenario in the presence of mismatch with $\bPi\neq \Ibb_n$. 
	Estimating $\bPi$ and $\Wbb$ without any constraint is infeasible due to the large number of parameters. In addition to $\bPi$ being block diagonal, we assume that only a small fraction of mismatch occurs and hence  $\nmis= o(n)$.  However, we do not constrain $\bPi$ to be a permutation matrix and accommodate more complex mismatch patterns. For example, if $\bX$ and $\bY$ represent ICD-10 and ICD-9 codes respectively, $\bY_i$ may not be mapped to any single ICD-10 code but rather needs to be represented by a combination of multiple ICD-10 codes in $\Xbb$. We also allow some columns of $\bPi$ to be zero vectors, indicating that the corresponding unit of $\Xbb$ does not link to any response in $\Ybb$. 
	In the presence of mismatch, we assume that $\bY_i \mid \Xbb,i\in\Gsc$ are independent and follow
	\begin{equation}\label{withpi}
		f_{\text{vMF}}(\bY_i|\Xbb;~\kappa) = C_{p}(\kappa)\exp(\kappa \bmu_{\bPi,i}\trans  \bY_i) \quad \mbox{with}\quad \bmu_{\bPi,i} = \Wbb\trans(\bPi_{i\cdot}\Xbb)\trans \ , \ \Wbb\Wbb\trans = \Ibb_p
		\end{equation}
	and $\|(\bPi_{i\cdot}\Xbb)\trans \|_2=1$ to ensure that the mapped vector $(\bPi_{i\cdot}\Xbb)\trans$ remains on $\hypersphere$. 
	A necessary condition for $\|(\bPi_{i\cdot}\Xbb)\trans \|_2=1$ is $\frac{1}{\sqrt{n_k}}\leq \|\bPi_{i\cdot}\|_2\leq \frac{1}{\sigma_{n_k}(\Xbb_{[G_k,:]})},\text{ for all $i\in G_k$}$, which is shown in Lemma~\ref{remark:pi}.
	We further assume that $n>p>\max_{1\leq k\leq K}n_k$ and $\kappa\neq 0$.

	\subsection{Iterative spherical regression mapping (iSphereMAP)}\label{est}
	
	We propose an iterative spherical regression mapping (iSphereMAP) method to estimate the translation matrix $\Wbb$ and the mapping matrix $\bPi$. 
	Although the iSphereMAP procedure can iterate until convergence, we find that the estimators stabilize after three steps and hence focus on the three-step procedure. 
	In step I, we simply initiate $\bPi$ as $\bPihat\supone = \Ibb_n$ and obtain an estimate of $\Wbb$ as
	\begin{equation}\label{modelW0}
		\Wbbhat\supone = 
		\argmin{\Wbb: \Wbb\Wbb\trans  = \Ibb_p}{\| \Ybb_{[\mm(\bPihat\supone),:]}- \Xbb_{[\mm(\bPihat\supone),:]} \Wbb\|_F^2}
		= \argmin{\Wbb: \Wbb\Wbb\trans  = \Ibb_p}{\| \Ybb- \Xbb \Wbb\|_F^2}
		= \argmin{\Wbb: \Wbb\Wbb\trans  = \Ibb_p}{\ellhat_0(\Wbb)}. \end{equation}
	The degree of dissimilarity between $\bPihat\supone$ and the true $\bPi$ is of size $\nmis=|\Dsc(\Ibb, \bPi)|=n-|\mm(\bPi)|$. 
	Solving for $\Wbb$ in the optimization problem (\ref{modelW0}) is a well-known orthogonal Procrustes problem \citep[e.g.]{schonemann1966generalized,gower2004procrustes}, the solution to which is the polar decomposition of $\Xbb\trans \Ybb$ \citep[e.g.]{higham1986computing}:
	\begin{equation*}
		\Wbbhat\supone=\Usc(\Xbb\trans \Ybb), \ \mbox{where for any nonsingular matrix $\Abb_{p \times p}$, } \Usc(\Abb)=\Abb(\Abb\trans \Abb)^{-\half}.
		\end{equation*}

	In step II, we obtain an improved estimator of $\bPi$ by mapping the translated data, $\Ybb$ and $\Xbb\Wbbhat\supone$. Recall that $\bPi=\text{diag}\{\bPi^1,\dots,\bPi^K\}$, where the mapping matrix for the $k^{th}$ group, $\bPi^k$, is an $n_k\times n_k$ matrix. We estimate each $\bPi^k$ using a hard-thresholding procedure as follows. First, we compute an initial estimate $\Initialpi^k$ by the ordinary least squares (OLS) as 
	\begin{equation*}
		\Initialpi^k=
		\Ybb_{[G_k,:]} (\Xbb_{[G_k,:]}\Wbbhat\supone)\trans (\Xbb_{[G_k,:]}\Xbb_{[G_k,:]}\trans )^{-1} .
		\end{equation*}
	Then to obtain a sparse estimate of $\bPi$, we apply hard-thresholding to $\Initialpi = \diag\{\Initialpi^1, ..., \Initialpi^K\}$ allowing for \onetomany  correspondence within group. 
	Specifically, for each $i\in \Gsc$, let
	\begin{equation*}
		\beta_i=1-\max_{j:j\sim i}\; \cos(\bPi_{i\cdot},\Ibb_{j\cdot}), \ \widetilde{\beta}_i=1-\max_{j:j\sim i}\; \cos(\Initialpi_{i\cdot},\Ibb_{j\cdot}) \text{, and }
		\jtilde_i = \argmax_{j:j\sim i}\; \cos(\Initialpi_{i\cdot},\Ibb_{j\cdot}).
		\end{equation*}
	Intuitively, one minus cosine similarity corresponds to distance, and thus $\beta_i$ measures how distinguishable $\bPi_{i\cdot}$ is from $\Ibb_{j\cdot}$ which encodes a \onetoone mapping.
	We can see that $\beta_i=0$ if $\bPi_{i\cdot}=\Ibb_{j\cdot}$ for some $j\sim i$, and $\beta_i \ne 0$ if $\bPi_{i\cdot}$ represents a  \onetomany  mapping. Thus, the
	support $\Csc =  \{i\in\Gsc: \beta_i\ne 0\}$ 
	indexes the rows where $\bPi_{i\cdot}$ corresponds to \onetomany  mapping. 
	To recover $\Csc$ and construct a sparse estimate of $\bPi$, denoted as $\bPihat\suptwo$, we threshold $\betatilde_i$ with a properly chosen $\lambda_n$ and obtain the $i^{th}$ row of $\bPihat\suptwo$ as
	\begin{equation}
		\bPihat\suptwo_{i\cdot}=   \Ibb_{\jtilde_i\cdot} \mathbbm{1}(\widetilde{\beta}_i\leq\lambda_n) + 
		\frac{\Initialpi_{i\cdot}}{\|(\Initialpi_{i\cdot}\Xbb)\trans\|_2} \mathbbm{1}(\widetilde{\beta}_i > \lambda_n) 
		\end{equation}
	where we suppressed $\lambda_n$ in $\bPihat\suptwo$ for ease of notation. 
	Thus, we set $\bPihat\suptwo_{i\cdot}$ to $\Ibb_{\jtilde_i\cdot}$ when $\widetilde{\beta}_i$ is small; but estimate $\bPi_{i\cdot}$ as $\Initialpi_{i\cdot}/\|(\Initialpi_{i\cdot}\Xbb)\trans\|_2$  when $\widetilde{\beta}_i$ is large. The $\ell_2$-normalized 
	estimator $\Initialpi_{i\cdot}/\|(\Initialpi_{i\cdot}\Xbb)\trans\|_2$ preserves unit length for the translated vector $(\Initialpi_{i\cdot}\Xbb)\trans$ and in fact is the solution to minimizing the constrained OLS problem under the spherical constraint.
	
	With a properly chosen $\lambda_n$, $\bPihat\suptwo$ consistently recovers $\bPi$ as detailed in Section~\ref{sec:thm_pi}. Intuitively,  to correctly classify $\bPi_{i\cdot}$ as a  \onetoone  or \onetomany  mapping, $\lambda_n$ should be
	chosen to be both below the smallest  non-zero signal of ${\beta}_i$ and above the estimation error of the zero-signals. 
	In practice, $\lambda_n$ is selected among a series of values in $(0,1-\frac{1}{\sqrt{2}})$ by cross-validation, where the upper bound was chosen because there is at most one $j$ that gives $\cos(\Initialpi_{i\cdot},\Ibb_{j\cdot})> \frac{1}{\sqrt{2}}$. Specifically, we use cross-validation optimizing the mean squared error for prediction of $\Ybb$, defined as $\sum_{cv}\|\Ybb_{cv}-\bPihat\suptwo\Xbb_{cv}\widehat{\Wbb}\|_F^2$, where $\Ybb_{cv}$ and $\Xbb_{cv}$ denote the combination of selected columns of $\Ybb$ and $\Xbb$, respectively, which serve as validation data.
	
	In step III, based on the updated mapping estimate $\bPihat\suptwo$, we obtain a refined estimator for $\Wbb$ using the subsample that we estimate to be correctly matched as
	\begin{equation*}
		\widehat{\Wbb}\suptwo=\Usc\left(\Xbb_{\ii}\trans \Ybb_{\ii}\right), \quad \mbox{where $\mm(\bPihat\suptwo)=\{i \in \Gsc: \bPihat\suptwo_{i\cdot}=\Ibb_{i\cdot} \}$. }
		\end{equation*}
		We detail the implementation of the above three-step iSphereMAP algorithm in Section~\ref{supp:ispheremap} of the Supplementary Material.
		Although the proposed algorithm can be iterated in practice, we show in the next section that $\Wbb$ and $\bPi$ can both be consistently estimated in three steps. 
	
	\section{Theoretical Properties of iSphereMAP Estimators\label{sec:theory}}
	\subsection{Properties of the initial translation matrix estimator $\Wbbhat\supone$}
	We first investigate whether $\Wbbhat\supone$ from the initial spherical regression (\ref{modelW}) can consistently estimate $\Wbb$ despite the presence of mismatch in the data. Intuitively, if only a small fraction of the data is mismatched, the distortion in $\Wbbhat\supone$ due to mismatch may be negligible. 
	The following theorem presents the error bound of $\Wbbhat\supone$, which is proved in Section~\ref{app-theorem-one} of the Supplementary Material.
	
	\begin{theorem}\label{thm:w-hat-consistent}
		For any $t>0$, 
		if $\gamma_{\kappa,p}\sigma_p(\Xbb)^2>  t \sqrt{n(1-\gamma_{\kappa,p}^2)}+2\gamma_{\kappa,p} \nmis $,
		then with probability at least $1-1/t^2$, 
		\begin{equation*}
			\|\Wbbhat\supone-\Wbb\|_F
			\leq \frac{t \sqrt{n(1-\gamma_{\kappa,p}^2)}+2\gamma_{\kappa,p} \nmis}{\gamma_{\kappa,p}\sigma_p(\Xbb)^2 - t \sqrt{n(1-\gamma_{\kappa,p}^2)}- 2\gamma_{\kappa,p} \nmis}.
			\end{equation*}
		
	\end{theorem}
	\begin{remark}\label{remark:sigma_p}
		The quantity $\sigma_p(\Xbb)$ describes the colinearity of columns of $\Xbb$, with a larger value suggesting less linearly dependent rows. If $p>n$, $\sigma_p(\Xbb)=0$.
		When $n\geq p$ and rows of $\Xbb$ are stochastically generated with a uniform distribution over the surface of the hypersphere $\hypersphere$, $\sigma_p(\Xbb)$ is roughly of the order $O(\sqrt{n/p})$ as $n$ and $p$ grow. This rate decreases as $p$ increases, mainly because of the spherical assumption that rows of $\Xbb$ are of unit length.
	\end{remark}
	\begin{remark}\label{remark:kappa}
		The error bound in Theorem~\ref{thm:w-hat-consistent} also depends on the scaling factor $\gamma_{\kappa,p} \in (0, 1)$ introduced in Section~\ref{vMFdefine}. In fact, the term 
		\begin{equation*}\eta_{\kappa,p} \equiv 1-\gamma_{\kappa,p}^2\end{equation*}
		describes the inherent noise in the data, with $ E[\|\Xbb^\top(\Ybb-E[\Ybb])\|_F^2]={n(1-\gamma_{\kappa,p}^2)}$ and  $E[(\Ybb-E[\Ybb])(\Ybb-E[\Ybb])^\top]=(1-\gamma_{\kappa,p}^2)\Ibb_n$.
		The noise level $\eta_{\kappa,p}$, determined by the order of $p$ and $\kappa$, drives the precision of the iSphereMAP estimators.
		In particular, if $p$ and $\kappa$ are fixed, then $\eta_{\kappa,p}$ is a positive constant with $\eta_{\kappa,p}\in(0,1)$. The larger $\kappa$ is, the more concentrated the data is around $\bmu$, the closer 
		$\eta_{\kappa,p}$ is to $0$. If $p/\kappa= o(1)$ and $p\geq 4$, then $\eta_{\kappa,p}\to 0$ as $\kappa \to \infty$ by Lemma~\ref{lemma:eta-bound}. One can interpret the two scenarios of $p$ and $\kappa$ as noisy and approximately noiseless in analogy to the Gaussian setting.
	\end{remark}
	The following corollary simplifies the error bound of $\Wbbhat\supone$ in the scenarios when $\eta_{\kappa,p}$ is a fixed constant or goes to zero as discussed in Remark~\ref{remark:kappa}, which is proved in Section~\ref{app-coro-one} of the Supplementary Material. The conditions required to achieve consistency is weaker than that in \cite{chang1986spherical}.
	
	\begin{corollary}\label{coro:w-hat-large-kappa}
		Suppose $\gamma_{\kappa,p}>\boundgamma$ for some constant $\boundgamma\in(0,1)$ that does not depend on $\kappa$ and $p$, $n\to\infty$, and $\nmis=o(\sigma_p(\Xbb)^2)$. Then we have
		\begin{equation}
			\begin{split}
				\|\Wbbhat\supone-\Wbb\|_F & = \left\{ \begin{array}{ll}
					O_P\left(\frac{\sqrt{n}+\nmis}{\sigma_p(\Xbb)^2}\right) & \mbox{if $p$ and $\kappa$ are fixed, $\sqrt{n}=o(\sigma_p(\Xbb)^2)$}\\
					O_P\left(\frac{\sqrt{n\eta_{\kappa,p}}+\nmis}{\sigma_p(\Xbb)^{2}} \right) & \mbox{if $\sqrt{n\eta_{\kappa,p}}=o( \sigma_p(\Xbb)^2)$.}
				\end{array} \right.\label{eq:wrate_pkinf}
			\end{split}\end{equation}
		In particular, $\|\Wbbhat\supone-\Wbb\|_F$ converges to $0$ in probability in both cases.
	\end{corollary}
	\begin{remark}
		When $p/\kappa=o(1)$, $\kappa\to\infty$, $p\geq 4$, we have that $\eta_{\kappa,p}=O(p/\kappa)$. In this case $\sqrt{n\eta_{\kappa,p}}=o( \sigma_p(\Xbb)^2)$ if $\sqrt{np/\kappa}=o( \sigma_p(\Xbb)^2)$. Thus we can consistently recover $\Wbb$ as long as the rate at which $\sigma_p(\Xbb)$ grows is faster than both $\nmis$ and $\sqrt{np/\kappa}$. In addition, note that $\sigma_p(\Xbb)\leq\|\Xbb\|_F=\sqrt{n}$. Therefore $\nmis=o(\sigma_p(\Xbb)^2)$ indicates $\nmis=o(n)$.
	\end{remark}
	\begin{remark}\label{remark2} 
		Assuming $\sigma_p(\Xbb)=O(\sqrt{n/p})$ as described in Remark~\ref{remark:sigma_p}, we can see from Corollary~\ref{coro:w-hat-large-kappa} that as $n\to\infty$, 
		$\|\Wbbhat\supone-\Wbb\|_F=o_P(1)$ under either of the following asymptotic regimes: (1) $p$ and $\kappa$ are fixed and $\nmis=o(n)$; or (2) $\kappa\to\infty$, $p\geq 4$, $p=o(\kappa)$, $\nmis=o(n/p)$
		and $p^3=o(n\kappa)$. 
	\end{remark}

	\subsection{Properties of the Mapping Matrix estimator} \label{sec:thm_pi}
	Since the mapping matrix estimator $\bPihat\suptwo$ is a thresholded version of the initial OLS estimator  $\Initialpi =\diag\{\bPitilde^1, ..., \bPitilde^{K}\}$, we first establish the convergence rate for $\bPitilde^k$ in the following theorem.
	
	\begin{theorem}\label{thm:thm-OLS-pi}
		If $n\to\infty$, $p\geq 4$, $n>p>\max_{1\leq k\leq K}n_k$, $\gamma_{\kappa,p}>\boundgamma $ for some constant $\boundgamma\in (0,1)$ that does not depend on $\kappa$ and $p$, $\sqrt{n\eta_{\kappa,p}}=o( \sigma_p(\Xbb)^2)$, and $\nmis=o(\sigma_p(\Xbb)^2)$, then
		\begin{equation}\label{onegrpPI}
			\|\Initialpi^k-\bPi^{k}\|_F
			= O_p\left(
			\sigma_{n_k}(\Xbb_{[G_k,:]})^{-1}\sqrt{n_k}\left\{
			\sqrt{\frac{p}{\kappa}}+\frac{\sqrt{n\eta_{\kappa,p}}+\nmis}{\sigma_p(\Xbb)^{2}})
			\right\}
			\right),
			\end{equation}
		for $k=1,\dots,K$. In addition, assume that $K\to\infty$ and $4\log K\leq p\min_{1\leq k\leq K}n_k$. Then,
		\begin{equation}\label{allgrpPI}
			\max_{1\leq k\leq K}\|\Initialpi^k-\bPi^{k}\|_F
			= O_p\left(
			[\min_{1\leq k\leq K}\sigma_{n_k}(\Xbb_{[G_k,:]})]^{-1}\max_{1\leq k\leq K}\sqrt{n_k}\left\{
			\sqrt{\frac{p}{\kappa}}+
			\frac{\sqrt{n\eta_{\kappa,p}}+\nmis}{\sigma_p(\Xbb)^{2}}
			\right\}
			\right).
			\end{equation}
	\end{theorem}
	\begin{remark}\label{remark1}
		The term $\min_{1\leq k\leq K}\sigma_{n_k}(\Xbb_{[G_k,:]})$ indicates the within group variation of the design matrix rows. In particular, if we assume that the pairwise cosine similarity within each group is no greater than $a$ where $a\leq \frac{1}{\max n_k-1} $, then $\min_{1\leq k\leq K}\sigma_{n_k}(\Xbb_{[G_k,:]})\geq 1-(\max n_k-1)a$.
	\end{remark}
	\begin{remark}\label{remark12}
		If the number of groups $K$ is fixed, then derivation from (\ref{onegrpPI}) to (\ref{allgrpPI}) is trivial. Our result concerns the nontrivial scenario when $K\to\infty$, in which case proof of Equation~(\ref{allgrpPI}) requires specific analysis of the tail bound behavior of the vMF distribution detailed in Proposition~\ref{lemma:tail-bound-von-miss} of the Supplementary Material.
	\end{remark}
	\begin{remark}\label{regimelambda}
		We discuss the asymptotic regime required by Theorem~\ref{thm:thm-OLS-pi} for the case where all groups have equal group size with $n_k\equiv n/K$, $\kappa\to\infty$, $p\geq 4$, $p=o(\kappa)$, and $\sigma_p(\Xbb)$ is of the order $\sqrt{n/p}$ as described by Remark~\ref{remark:sigma_p}. First, $p$ needs to be small enough compared to $n$, $\kappa$ and $n\kappa$ ($p=o(\kappa)$ and $p=o(n^{1/3}\kappa^{1/3})$ by Remark~\ref{remark2}, and $p<n$) so that the error rate of $\Wbbhat\supone$ is controlled by Corollary~\ref{coro:w-hat-large-kappa}. Second, $p$ needs to be larger than $n_k\equiv n/K$ so that the OLS has a unique solution. Third, the mismatch needs to be sparse enough such that $\nmis=o(n/p)$ by Remark~\ref{remark2}.
		In summary, suppose $p=n^{r_1}$, $\kappa=n^{r_2}$, and $K=n^{r_3}$, then the conditions of Theorem~\ref{thm:thm-OLS-pi} are satisfied when  $0<r_3<1$, $1-r_3<r_1<\min(1,(1+r_2)/3,r_2)$, and $\nmis=o(n^{1-r_1})$.
	\end{remark}

	Interpretation of Theorem~\ref{thm:thm-OLS-pi} is relatively straightforward. The origin of the error in the initial OLS estimate of $\bPi^k$ is four-fold. 
	First, the inherent error of the vMF distribution contributes the term $\sqrt{p/\kappa}$. This is a unique tail bound property of the vMF distribution which we derive in Proposition~\ref{lemma:tail-bound-von-miss} of the Supplementary Material. In particular, when $p$ is fixed, or $p=o(\kappa)$, then as the concentration parameter $\kappa$ goes to infinity, the data approach the noiseless situation and this term goes to zero.
	Second,  by Corollary~\ref{coro:w-hat-large-kappa}, the estimation error of $\Wbb$ in the previous step contributes the term $\sigma_p(\Xbb)^{-2}(\sqrt{n\eta_{\kappa,p}}+\nmis)$. Third, the error bound of $\Initialpi^k$ is proportionally dependent on the size of $\bPi^k$. Lastly, if two rows within the same group have cosine similarity approaching one, then they are indistinguishable. Accordingly, the error bound is also scaled by the separability of rows in the design matrix $\Xbb_{[G_k,:]}$ as discussed in \textit{Remark}~\ref{remark1}. The proof of Theorem~\ref{thm:thm-OLS-pi} can be found in Section~\ref{app-theorem-two} of  the Supplementary Material.

	With the additional thresholding step, $\bPihat\suptwo$ attains model selection consistency as summarized in the following theorem, which is proved in Section~\ref{app-theorem-three} of the Supplementary Material.
	\begin{theorem}\label{thm:thresholding}
		Suppose that the assumptions in Theorem~\ref{thm:thm-OLS-pi} hold. Let $\Bsc_{\min} = \min_{i\in \mathcal{C}} \beta_i$ and 
		\begin{equation*}
			c_n=[\min_{1\leq k\leq K}\sigma_{n_k}(\Xbb_{[G_k,:]})]^{-1}\max_{1\leq k\leq K}\sqrt{n_k}\left\{
			\sqrt{\frac{p}{\kappa}}+\sigma_p(\Xbb)^{-2}(\sqrt{n\eta_{\kappa,p}}+\nmis)
			\right\}.
			\end{equation*}
		We further assume that $c_n\max_{1\leq k\leq K}\sqrt{n_k}\ll \Bsc_{\min}^2$, and $c_n\max_{i\in\mathcal{C} }\|\bPi_{i\cdot}\|_2\max_{1\leq k\leq K}\sqrt{n_k}\to 0$.
		Then, for $c_n \ll \lambda_n\ll \Bsc_{\min}$,  as $n\to \infty$,  the following holds with probability approaching one
		\begin{equation*}
			\begin{split}
				&\text{for all }i \in\mathcal{C},\;\max_{i\in\mathcal{C}}\|\bPihat\suptwo_{i\cdot}-\bPi_{i\cdot}\|_2 \to 0\\
				&\text{for all }i \notin\mathcal{C},\; \bPihat\suptwo_{i\cdot}=\bPi_{i\cdot}=\Ibb_{j\cdot} .
			\end{split}  
			\end{equation*}
	\end{theorem}
	
	Theorem~\ref{thm:thresholding} states that, as $n$ increases, our hard-thresholding procedure can distinguish between \onetoone  and \onetomany  mapping, correctly locate the matched row for \onetoone mapping, and consistently estimate the weight vector for \onetomany  mapping. 
	
	\begin{remark}
		The model selection consistency in Theorem~\ref{thm:thresholding} requires $p/\kappa=o(1)$, under which the noise level $\eta_{\kappa,p}=E[\|\bY_i-E[\bY_i]\|_2^2]=o(1)$.  Although not directly comparable, a similar condition was required in \cite{pananjady2016linear} where they assumed the following univariate linear regression $\bY=\bPi \Xbb\bw+\bU$, with $\bPi$ being a permutation matrix and $\bX$ being Gaussian. They studied the maximum likelihood estimate of $\bPi$ with the restriction of $\bPi$ being a permutation matrix. They showed that exact permutation recovery requires that the signal-to-noise ratio goes to infinity at a polynomial order of $n$. We require the noise level $\eta_{\kappa,p}=o(1)$ but do not require a specific rate.
	\end{remark}

	\begin{remark}\label{remark:lambda}
		To provide some intuition for the choice of $\lambda_n$, we note that 
		if $i\in\mathcal{C}$, i.e., the true $\bPi_{i\cdot}$ in fact represents a \onetomany  mapping, then ${\beta}_i\neq 0$. Thus $\lambda_n$ should be chosen to be much smaller than the smallest non-zero signal $\Bsc_{\min}$. On the other hand, if $i\notin \mathcal{C}$, then $\beta_i = 0$ and $\lambda_n$ should be able to tolerate the error in the initial estimate $\Initialpi$ and correctly threshold $\widetilde{\beta}_i$ to zero. The lower bound $c_n$ represents the order of $\max_{1\leq k \leq K}\|\Initialpi^k-\bPi^k\|_F$ by Theorem~\ref{thm:thm-OLS-pi}. By letting $\lambda_n \gg c_n$, we would successfully set the corresponding $\betatilde_i$ to zero. 
		If $\Xbb$ is uniformly distributed on the sphere and $n_k\ll p$,  $\min_{1\leq k\leq K}\sigma_{n_k}(\Xbb_{[G_k,:]})$ is approximately constant rate. Under the asymptotic regime of Remark~\ref{regimelambda} we have
			$c_n = O(\sqrt{n^{r_1-r_2}})$,
			where $r_1-r_2<\min(0,1-r_2,(1-2r_2)/3)$, with $r_1=\log(n)/\log(p)$ and $r_2=\log(n)/\log(\kappa)$. If we further assume that $\Bsc_{\min}$ is constant rate, then $\lambda_n$ needs to satisfy
			$\sqrt{n^{r_1-r_2}}\ll \lambda_n \ll 1$.
	\end{remark}

	\subsection{Properties of the Refined translation matrix estimator $\Wbbhat\suptwo$ \label{refineW}} 
	From Corollary~\ref{coro:w-hat-large-kappa}, the error bound of the initial estimate $\widehat{\Wbb}\supone=\Usc(\Xbb\trans \Ybb)$ consists of two terms of order $\sigma_p(\Xbb)^{-2}\nmis$ and  $\sigma_p(\Xbb)^{-2}\sqrt{n\eta_{\kappa,p}}$ respectively, with the first term accounting for the mismatch error. If $\bPihat\suptwo$ accurately identifies the mismatch patterns, then one would expect $\Wbbhat\suptwo$ to have lower error due to the removal of the mismatched pairs in Step III. The following corollary summarizes the error rate of $\Wbbhat\suptwo$, which is proved in Section~\ref{app-coro-two} of the Supplementary Material. 
	
	\begin{corollary}\label{coro:refinement}
		Under the assumptions of Theorems~\ref{thm:thm-OLS-pi} and \ref{thm:thresholding}, as $n\to\infty$
		we have 
		\begin{equation*}
			\|\Wbbhat\suptwo-\Wbb\|_F
			= O_P\left(\frac{\sqrt{(n-\nmis)\eta_{\kappa,p}}}{\sigma_p(\Xbb_{[\mm(\bPi),:]})^{2}} \right) = O_P\left(\frac{\sqrt{n\eta_{\kappa,p}}}{\sigma_p(\Xbb)^{2}} \right). 
			\end{equation*}
	\end{corollary}
	\begin{remark}
		Since $\nmis=o(n)$ is a necessary condition as discussed in Remark~\ref{remark2},  $n-\nmis$ is of the same order as $n$ . In addition, $\sigma_p(\Xbb_{[\mm(\bPi),:]})^{2}$ and $\sigma_p(\Xbb)^{2}$ are of the same order when $\nmis=o(\sigma_p(\Xbb)^{2})$, which is shown in Section~\ref{app-coro-two} of the Supplementary Material. 
	\end{remark}
	
	\begin{remark}
		Corollary~\ref{coro:refinement} indicates that estimating $\Wbb$ using only pairs deemed as matched by $\bPihat\suptwo$ reduces the error due to mismatch at the cost of reduced sample size  $n-\nmis$. However, since $\nmis = o(n)$, $\Wbbhat\suptwo$ attains the same error rate as the estimator obtained with $\bPi$ given or $\bPi = \Ibb$. That is, the iSphereMAP estimator $\Wbbhat\suptwo$ achieves an error rate that is as good as if no mismatch is present. Moreover, compared to the error rate of $O_P\{(\sqrt{n\eta_{\kappa,p}}+\nmis)/\sigma_p(\Xbb)^{2} \}$ in (\ref{eq:wrate_pkinf}), $\Wbbhat\suptwo$ attains a lower error rate than that of $\Wbbhat\supone$ when $\sqrt{n\eta_{\kappa,p}} = o(\nmis)$.
	\end{remark}

	\section{Simulation}\label{simu}
	We conduct extensive simulation studies to evaluate the performance of our proposed iSphereMAP method for estimating both $\Wbb$ and $\bPi$ and to compare to the  \cite{mikolov2013exploiting} approach, referred to as the MT method hereafter. Specifically, for each $i$, the MT method finds $j_i=\arg\max_{j}\cos(\bY_i,\widehat{\Wbb}\bX_{j})$ without using group information, where $\widehat{\Wbb}$ is obtained from the OLS. We compare 
	(1) estimates of $\Wbb$ from our proposed spherical regression and from OLS, using full data and refined data; (2) estimates of $\bPi$ from the hard-thresholding procedure using group information, and from the MT method without group information.
	
	Throughout our simulation, we set $p=300$, $\kappa=150$, and all results are averaged over $100$ simulation datasets. This is a scenario where the noise level is much higher than the theoretical settings. 
	For a given sample size $n$, we let the true mapping matrix $\bPi$ include $\nmis=n^{\alpha}$ mismatched rows. We fix $n = 8000$ with $\alpha$ ranging from 0.35 to 0.93, corresponding to 0.3\% to 53\% of mismatched pairs among the entire data. We also fix $\alpha = 0.8$ but with $n$ varying from approximately 2000 to 8000. The sample size $n$ increases as the number of groups $K$ increases. Specifically, we prespecify a list of $1700$ unequal group sizes. We select the first $K$ group sizes in the list, with $K$ ranging from $100$ to $1700$, such that $n$ increases from approximately $2000$ to $8000$. 
	With a specific set of $(K,n,\alpha)$, we first simulate $\Xbb$ by generating $n$ vectors that follow a mixture of $K$ vMF distributions with concentration parameter $\kappa$, whose mean directions are $K$ group centers uniformly distributed on $\hypersphere$. The mixture weight for the distribution of the corresponding group is twice the weight for the other $K-1$ distributions.
	Then we generate $\bPi=\text{diag}\{\bPi^1,\dots,\bPi^K\}$, in which randomly selected $n-n^\alpha$ rows are copied from the corresponding rows of $\Ibb_n$, whereas the other $n^\alpha$ rows are specified to encode \onetoone  and \onetomany  mismatch patterns. We let half of the $n^\alpha$ rows be indicators that introduce permutation within group and the other half be weight vectors following the Uniform(0,1) distribution to introduce \onetomany  mapping.
	We specify the true transformation matrix $\Wbb$ by taking the left eigenvectors of a $p\times p$ matrix of standard normal random values. 
	Finally, we generate $\Ybb$ with mean directions $\bPi \Xbb \Wbb$ following the vMF distribution with concentration parameter $\kappa$.
	
	\begin{figure}[ht]
		\begin{centering}
			\makebox[\textwidth][c]{\includegraphics[width=0.9\textwidth, scale=.5]{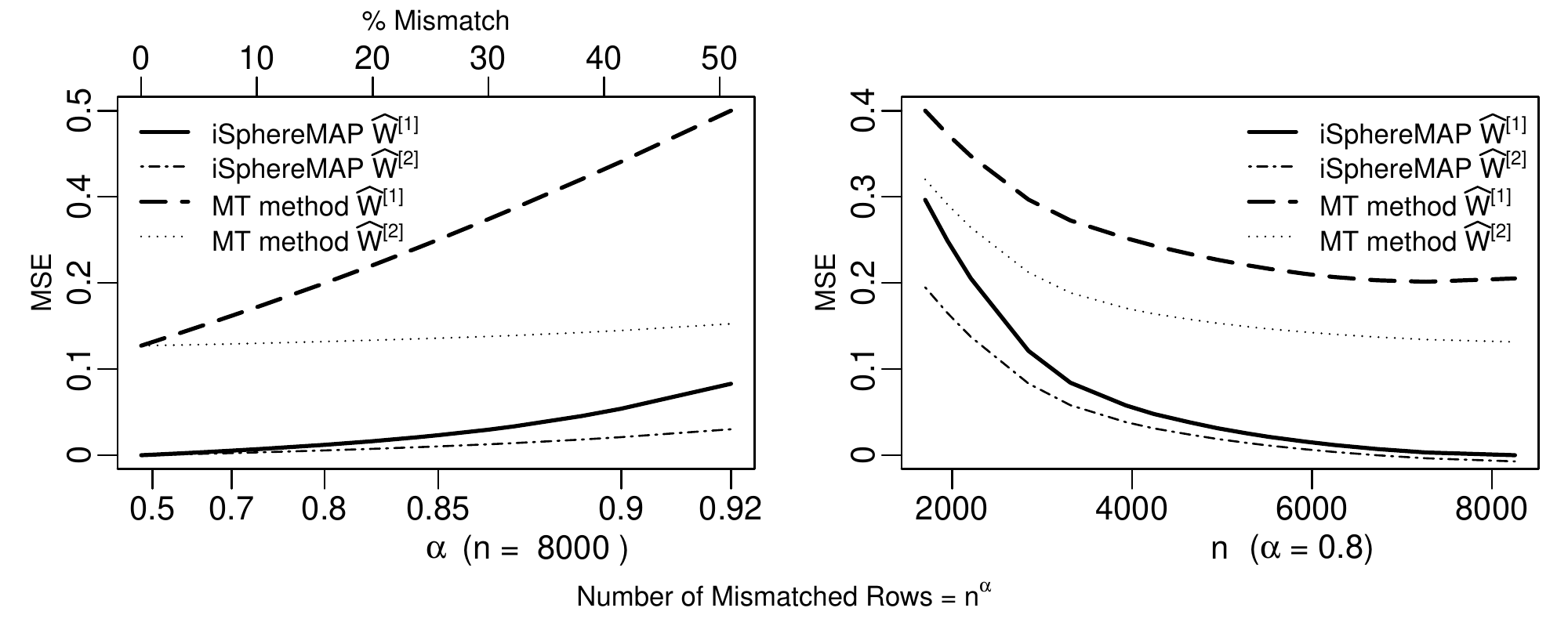}}
			\par \end{centering}
		\caption{\label{Wsimurslt} Performance of $\Wbbhat\supone$ and $\Wbbhat\suptwo$ obtained based on the proposed spherical regression and OLS in terms of the MSE (normalized by $p^{-1}=1/300$) under ranging amount of mismatch (left panel) and sample size (right panel).}
	\end{figure}
	
	We first summarize in Figure~\ref{Wsimurslt} the mean squared errors (MSEs) scaled by $p^{-1}$ of $\Wbbhat\supone$ and $\Wbbhat\suptwo$ from spherical regression and the MT method (OLS). The MSE is defined as the average of $\|\Wbbhat-\Wbb\|^2_F$ over simulated datasets.
	The spherical regression attains considerably smaller estimation error compared to the MT method in both $\Wbbhat\supone$ and $\Wbbhat\suptwo$. As $\alpha$ and correspondingly $\nmis$ increases, both methods suffer increased error as expected but the deterioration is much more drastic for the OLS. For a fixed $\alpha$, the estimation error of spherical regression approaches to zero at a much faster rate than that of the OLS as $n$ increases. 
	
	By removing the unmatched pairs, substantial improvement is observed in $\Wbbhat\suptwo$ compared to $\Wbbhat\supone$. In particular, when $\nmis = n^{\alpha}$ ranges from $n^{0.7}$ to $n^{0.93}$, the MSEs from both methods are notably smaller than that of the initial estimates. Our observation is consistent with our discussion in Section~\ref{refineW} that when the order of $\nmis$ is larger than $\sqrt{n\eta_{\kappa,p}}=n^{0.5}$, the error rate of the refined estimate will be improved. With $\alpha$ fixed and $n$ increasing, $\Wbbhat\suptwo$ also have a consistently smaller MSE than $\Wbbhat\supone$, with the difference in MSE between $\Wbbhat\supone$ and $\Wbbhat\suptwo$ from spherical regression decreasing as $n$ increases.

	We next evaluate the performance of $\bPihat\suptwo$ obtained using data $(\Wbbhat\supone\Xbb,\Ybb)$ with and without the aid of group information, where $\Wbbhat\supone$ is obtained from the spherical regression.  
	Note that without a group structure, initial OLS estimate $\Initialpi$ may not be obtained due to the high dimensionality. In this case, we estimate a permutation matrix using the MT method which matches rows of $\Wbbhat\supone\Xbb$ and $\Ybb$ using cosine similarity as distance metric.
	We evaluate both the \onetoone  match rate and the MSE of \onetomany  weight defined as follows. 
	The \onetoone  match rate is the percentage of correctly matched rows among all \onetoone  mappings. Specifically, we calculate the \onetoone  match rate as
	$|\{i:\widehat{\bPi}_i=\bPi_i, i\in\Csc^c\}|/|\Csc^c|$, where $\Csc^c$ is the complement of $\Csc$, i.e., the true index set of \onetoone  mapping.
	The MSE of \onetomany  weight is defined as the MSE of $\bPihat_{[\Csc,:]}$ normalized by its size $|\Csc|n$. We also access the percentage of correctly identified \onetomany  mappings, i.e., $\widehat{\Csc}\cap\Csc|/|\Csc|$, where $\widehat{\Csc}$ denotes the estimated set of \onetomany  mapping.
	
	\begin{figure}[ht]
		\begin{centering}
			\makebox[\textwidth][c]{\includegraphics[width=0.9\textwidth, scale=.5]{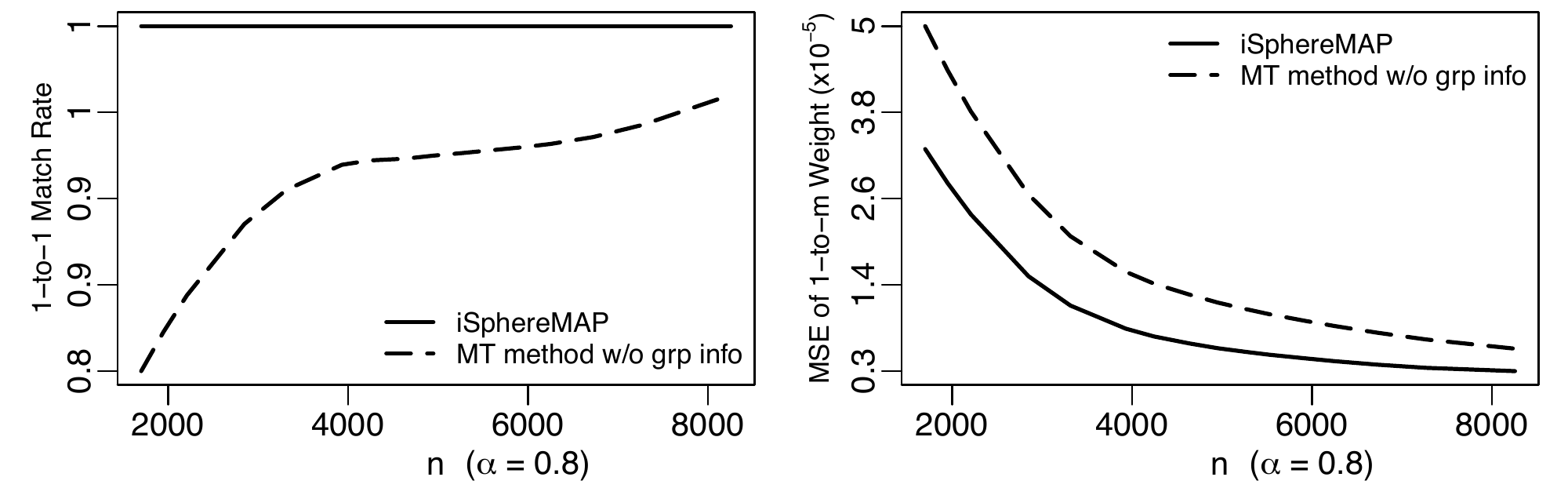}}
			\par \end{centering}
		\caption{\label{Pisimurslt} Performance of $\bPihat\suptwo$ obtained with and without group information in terms of the \onetoone  match rate (left panel) and the MSE of \onetomany  weight (right panel).}
	\end{figure}
	
	Figure~\ref{Pisimurslt} presents the performance of $\widehat{\bPi}$ obtained from our method with group information and from the MT method without group information, with a goal to understand the amount of accuracy gain from the group information. As $n$ increases, the match rate for \onetoone  mapping increases and the MSE of the weight vectors decreases. Our proposed method using group information outperforms the MT method without group structure in terms of both the \onetoone match rate and the MSE of \onetomany  mapping weight.
	Moreover, our proposed hard-thresholding procedure can correctly identify 95\% of the \onetomany  mappings on average across all scenarios, whereas the MT method does not allow for \onetomany  mapping.
	
			To further examine the robustness and efficiency of the iSphereMAP procedure, we performed simulation studies under three alternative scenarios:
			(I)  the block-diagonal structure of $\bPi$ is overly coarse; 
			(II) only \onetoone mapping is present; 
			and (III) a low noise level scenario compatible with the theoretical settings. Results from these scenarios are detailed  in Sections~\ref{wrong_grp_info}-\ref{supp:lessnoise} of the Supplementary Material.
			We observe that our method is not substantially sensitive to the overly coarse group structure. We thus generally recommend to be conservative in choosing the group structure. When only \onetoone mapping is present, our method remained better performance compared to the MT method, which is expected as the MT method is not customized to spherical data and does not utilize the group information. Lastly, with less noise in the data the estimators have relatively less MSE and match rate.
			
	\section{Application: ICD code Translation}\label{application}
	In this section, we employ the iSphereMAP method to (i) map the ICD-9 codes between two healthcare systems, the PHS and the VHA; and (ii) to automatically translate between ICD-9 and ICD-10 codes using VHA data. For the code mapping between healthcare systems, we focus on the ICD-9 codes since the majority of the codes recorded in the EHR are ICD-9 codes. In both examples, we use the \texttt{word2vec} algorithm to obtain SEVs for ICD codes  as detailed in Section~\ref{ICD-9motivation}. The code-SEVs are $\ell_2$-normalized.
	
	\subsection{Mapping ICD-9 codes between VHA and PHS\label{example1}} 
	The VHA is the largest integrated health care system of the united states, with an integrated EHR system adopted by all veterans hospitals and clinics \citep{VHA}.
		The PHS is a non-profit health care system founded by Brigham and Women's Hospital and Massachusetts General Hospital \citep{PHS}. The code SEVs for VHA were trained using data from about 18 million veterans. The PHS SEVs were trained using EHR data from about 62,000 patients that belong to the PHS Biobank cohort.
	There are a total of $n=8823$ ICD-9 code-SEVs each of dimension $p=300$ from the two systems available for analysis. Grouping information on the ICD codes is available through the ICD hierarchy \citep{world1977manual,icd_hierarchy}, the Clinical Classification Software \citep{CCS}, or the ICD-to-phenotype mapping provided by the PheWAS catalogue \citep{denny2010phewas}. 
	We chose the phenotype code (namely phecode) as it represents clinically meaningful phenotypes. Due to the hierarchical nature of the phecodes, we collapsed all phecodes with the same integer values into the same group, resulting in $K = 578$ groups. 
	The ICD-9 codes from different phecode groups represent distinct phenotypes and thus are unlikely to be confused with each other. As such, no mismatch is expected to occur across groups. On the other hand, we expect to see mismatch within groups. In fact, it has been shown that the level of agreement among coders and agencies in assigning medical codes for a specific disease or procedure can be poor \citep{austin2002multicenter,o2005measuring}, in part due to the fact that multiple codes can be appropriate for describing the same diagnosis. 
	The ICD-9 code SEVs trained from VHA and PHS data have been presented in Figure~\ref{grouping} of Section~\ref{ICD-9motivation}. 
	We can see that the code-vectors in VHA and PHS generally show distinct patterns, reflecting the variation in languages used in the two healthcare systems that necessitates alignment of the two language spaces. In addition, although the codes are clustered by the phecode group, many of the groups are distributed on top of each other, suggesting the difficulty in matching the codes without prior group information.

	\begin{figure}[h]
		\begin{subfigure}[t]{\textwidth}
			\centering\includegraphics[width=0.8\textwidth]{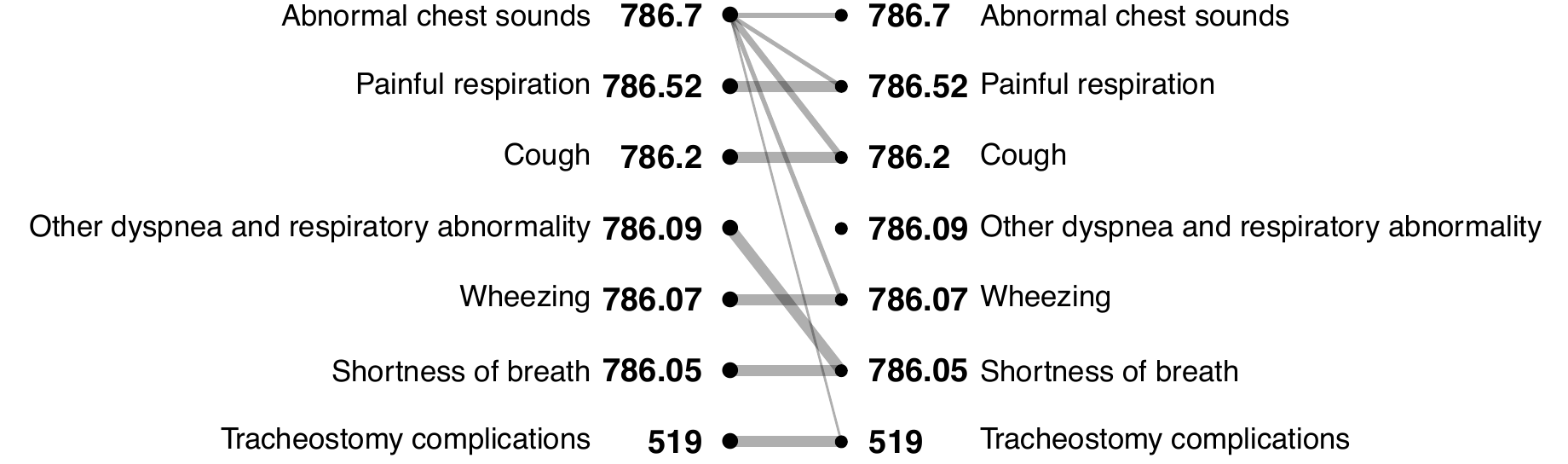}
			\caption{
				Symptoms of respiratory system}
		\end{subfigure}\vspace{0.1in}
		\begin{subfigure}[t]{\textwidth}
			\centering\includegraphics[width=0.8\textwidth]{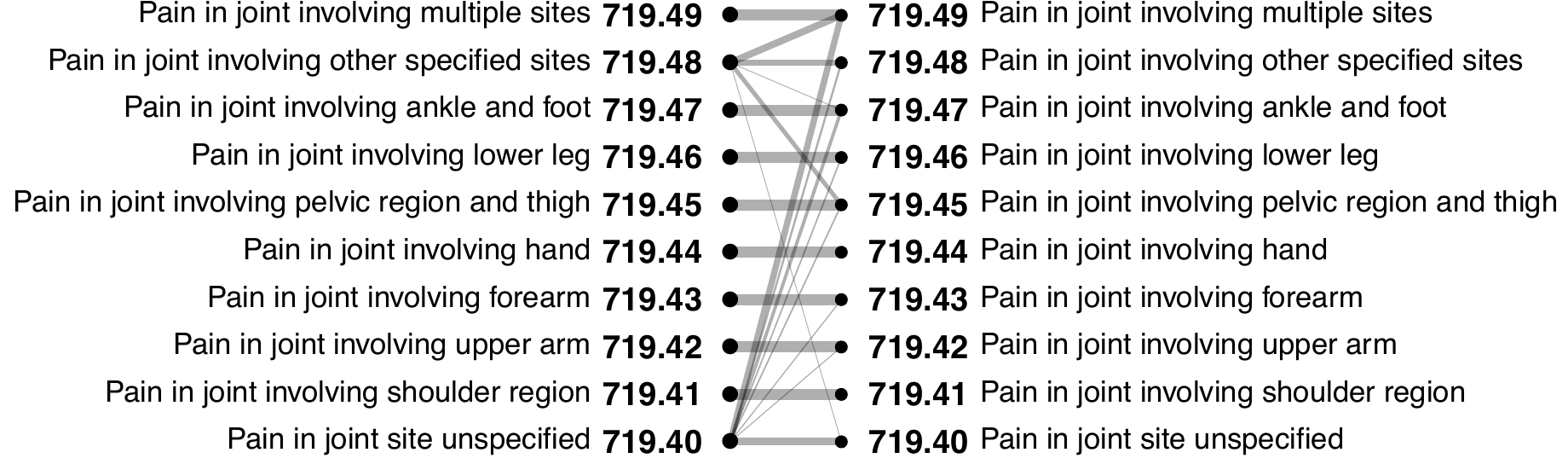}
			\caption{
				Pain in joint}
		\end{subfigure}
		\caption{\label{correspondence} Plot of the estimated mapping of codes from VHA (left) to PHS (right). Selected codes belong to the group describing (a) symptoms of respiratory system, and (b) pain in joint. Line width indicates the magnitude of weight vector components.}
	\end{figure}
	We select two groups of ICD-9 codes to present the result: one describing symptoms of respiratory system, the other describing pain in joint. 
	Figure~\ref{correspondence} presents the estimated mapping of codes from VAH (left) to PHS (right) from the iSphereMAP procedure. Thicker lines indicate larger weight for the corresponding codes on the right, and we do not link codes with negative weights. 
	In Figure~\ref{correspondence} (a), ICD-9 code 786.09 describing ``Other dyspnea and respiratory abnormality" is mapped to multiple codes with higher weights on both itself and code 786.05 describing ``shortness of breath", which is semantically similar to ``dyspnea". These two codes are likely to be used in an exchangeable manner.
	In Figure~\ref{correspondence} (b), most codes have a \onetoone  correspondence. However, codes 719.40 and 719.48 in VHA are mapped to multiple codes in PHS. Both codes describe joint pain with unspecified sites. It is thus reasonable to interpret these codes by combinations of codes associated with different specific sites or unspecified sites. The above observed patterns have been validated by domain experts.
	
	\subsection{Translation between ICD-9 and ICD-10 codes}\label{example2}
	We also apply our method to automatically map between ICD-9 and ICD-10 codes using VHA data. We train ICD-9 and ICD-10 code-vectors using data from non-overlapping time period, thus each set of vectors forms a language space. We take the GEM mapping \citep{GEM_2013} as a benchmark. As discussed in Section~\ref{ICD-9motivation}, due to the complexity and large number of ICD-10 codes, many mappings are \onetomany  or approximate match in GEM. For example, Figure~\ref{icd910fig} (a) displays the GEM mapping for ICD-9 codes in the rheumatoid arthritis (RA) group, which includes one-to-one, one-to-many, and many-to-one mappings and all are marked as ``approximate". 
	When an ICD-9 code should map to the combination of the corresponding ICD-10 codes according to the GEM mapping, e.g. ``714.2" in Figure~\ref{icd910fig} (a), we duplicate the ICD-9 code vector rows to match the number of ICD-10 codes to introduce mismatch error in the data.
	We define a group for pairs of GEM-linked ICD-9 and ICD-10 codes as one in which all ICD-9 codes have the same phecode up to the first decimal point to achieve moderate group sizes.
	Our final dataset includes $n=11025$ ICD-9 and ICD-10 SEV pairs ($p=600$) belonging to $K=1463$ groups, with 42\% \onetomany mapping and 58\% \onetoone mapping.
	\begin{figure}[!h]
		\begin{subfigure}[t]{\textwidth}
			\includegraphics[width=1\textwidth]{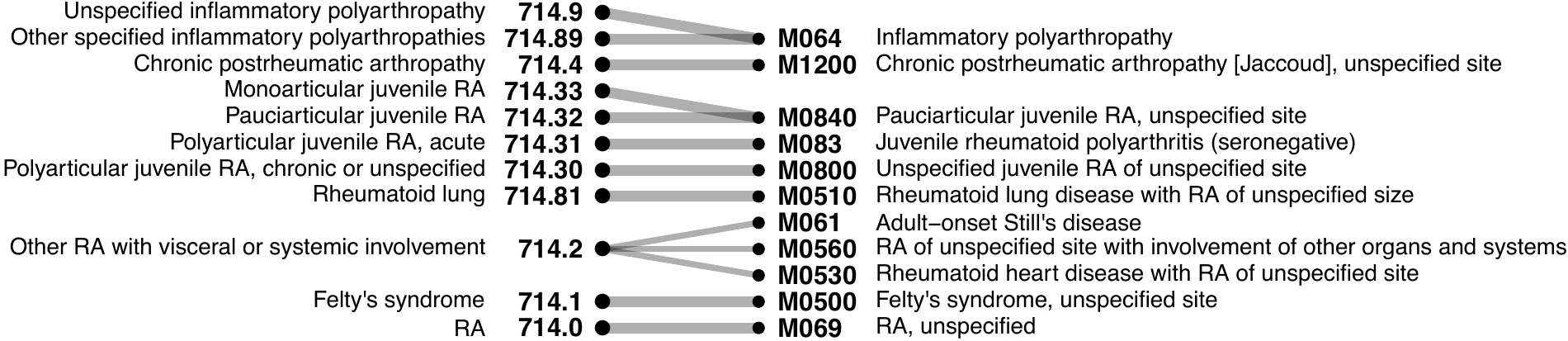}
			\caption{
				GEM ICD-9-to-10 mapping for rheumatoid arthritis (RA)
			}
		\end{subfigure}\vspace{0.1in}
		\begin{subfigure}[t]{\textwidth}
			\includegraphics[width=1\textwidth]{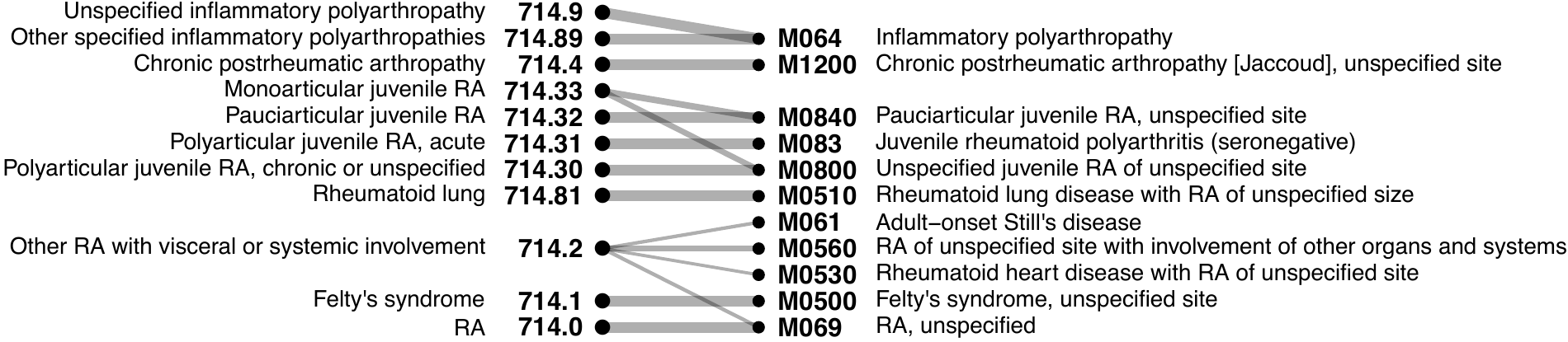}
			\caption{
				iSphereMAP estimated ICD-9-to-10 mapping for rheumatoid arthritis (RA)
			}
		\end{subfigure}\vspace{0.1in}
		\begin{subfigure}[t]{\textwidth}
			\includegraphics[width=1\textwidth]{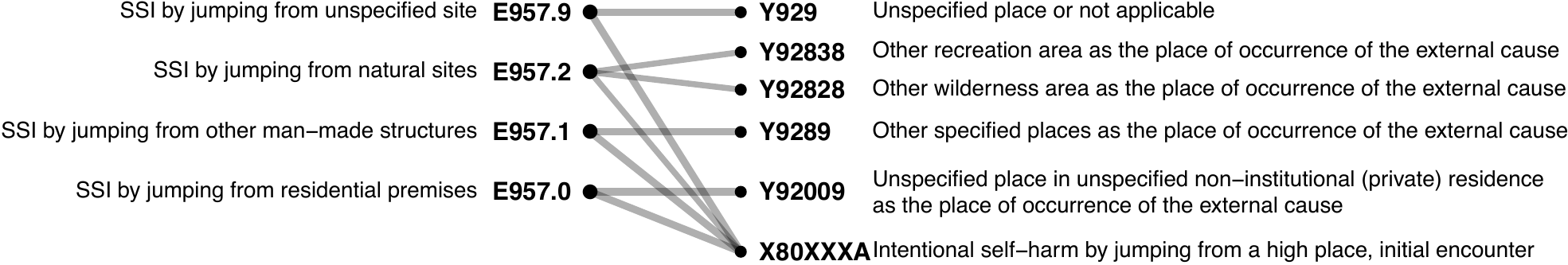}
			\caption{
				GEM ICD-9-to-10 mapping for suicide and self-inflicted injuries (SSI)
			}
		\end{subfigure}\vspace{0.1in}
		\begin{subfigure}[t]{\textwidth}
			\includegraphics[width=1\textwidth]{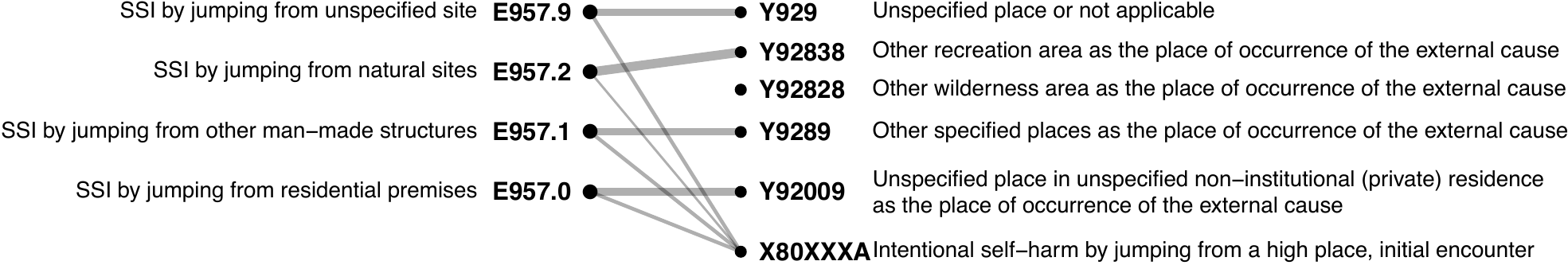}
			\caption{
				iSphereMAP estimated ICD-9-to-10 mapping for suicide and self-inflicted injuries (SSI)
			}
		\end{subfigure}
		\caption{\label{icd910fig} Plot of the manually created GEM mapping and data-driven mapping from ICD-9 to ICD-10 codes in the group describing rheumatoid arthritis (RA) and suicide and self-inflicted injuries (SSI). Line width indicates the magnitude of weight vector components.}
	\end{figure}
	
	Figure~\ref{icd910fig} (b) shows the estimated mapping from iSphereMAP, which is able to pick up different types of mapping patterns and only differs from the GEM mapping in very few codes. Additional interesting example of mapping for suicide and self-inflicted injuries (SSI) is presented in Figure~\ref{icd910fig} (c-d).
	We further investigate the proportions of correctly identified \onetoone and \onetomany mappings as well as the correctly matched code-pairs, taking the GEM mapping as the benchmark to validate the result. For comparison, we use the MT method with and without phecode-group structure in estimation of $\bPi$. 
	Using the phecode-group information, the MT method correctly matches 1298 (20\%) code-pairs among the 6359 code-pairs correctly identified as \onetoone mapping. Note that the MT method assumes that all mappings are one-to-one. However, without group information no code-pair can be correctly matched. 
	In contrast, our iSphereMAP correctly matches 2060 (49\%) code-pairs among the 4207 code-pairs correctly identified as \onetoone mapping. In addition, our method can further identify 54\% (2525) of the \onetomany mapping cases among 4666 \onetomany  mappings in total.

	\section{Discussion}\label{discussion}
	Data-driven semantic embeddings such as ICD code-SEVs are powerful approaches to learning the interpretation of medical codes in routine clinical practice which may differ when endorsed by different providers. We propose a novel code translation method with imperfectly linked embeddings by casting the translation problem into a statistical problem of spherical regression under mismatch. We detail the iSPhereMAP algorithm for estimating the translation matrix $\Wbb$ and the mapping matrix $\bPi$ and provide theoretical guarantees. In particular, we detail the extent of mismatch under which one may obtain a consistent estimate of $\Wbb$, and demonstrate that removing identified mismatched data based on the sparse estimate of $\bPi$ yields an improved estimator for $\Wbb$. In addition, we characterize conditions under which the support and magnitude of the mapping matrix $\bPi$ can be recovered. 
	Unlike existing methods in the literature on regression with mismatched data and machine translation, the  iSPhereMAP procedure allows for both \onetoone  and \onetomany  mapping, and can incorporate group structure when group information is available. Our method performs substantially better than methods limited to \onetoone correspondence and without using grouping information.
	Our methodological framework is particularly appealing because it can be extended to a wide range of applications, including confounding adjustment via text matching using text data in social science \citep{roberts2018adjusting,mozer2018matching}, and cross-language record linkage \citep{song2016cross,mcnamee2011cross}. The learned mapping matrix $\bPi$ and translation matrix $\Wbb$ have key practical value in transferring statistical models across systems \citep{torrey2010transfer}, capturing the pose of objects \citep{zhou2014vision}, estimating the relative angle of proteins \citep{sael2010binding} and so on.
	
		In the refined estimation of $\Wbb$, we only use data deemed correctly matched according to $\bPihat\suptwo$ to obtain $\Wbbhat\suptwo$. Removing mismatched data yields negligible information loss under the current setting of sparse mismatch with $\nmis=o(n)$. However, for settings with a large amount of mismatch, one may first use $\bPihat\suptwo$ to correct \onetoone mismatched data, and then estimate $\Wbb$ using all data that can be one-to-one mapped. As shown in a simulation study described in Section~\ref{supp:add_mismatch_1to1} of the Supplementary Material, further including corrected data can improve the performance, with improvement more substantial as the percentage of mismatch increases. Theoretical properties of such an alternative refinement strategy warrant further research, particularly for the setting where $\nmis/n \not\to 0$.
	
	The model selection consistency of $\bPi$ currently relies on an approximately noiseless condition where the noise level $\eta_{\kappa,p}=o(1)$, for which a sufficient condition is $\kappa\to\infty$, $p=o(\kappa)$, and $p\geq 4$. A similar condition that the signal-to-noise ratio goes to infinity was required in \cite{pananjady2016linear}. The seemingly stringent condition is in fact reasonable because in practice, normalization of the original data to unit length often substantially reduces the noise in the data. Our findings established a theoretical basis for future research on weaker conditions for mapping recovery.
	When the number of groups $K$ is relatively small such that some group size $n_k$ is larger than $p$, we may not be able to obtain an initial OLS estimate of $\bPi$. In this case, one may consider the alternative sparsity condition that $\|\bPi-\Ibb\|_1$ is small, under which shrinkage estimators such as the LASSO can be used to obtain $\Initialpi$. Modified iSphereMAP procedure under such settings warrants future research. In addition, although a fixed threshold was proposed to obtain a sparse estimator for $\bPi$, applying adaptive weights to allow the threshold to vary across groups and/or codes may further improve the performance. For example, a potential strategy is to adapt to the initial estimate $\Initialpi_{i\cdot}$ by measuring how distinguishable it is from a \onetoone mapping. As shown in Section~\ref{supp:adaptive_threshold} of the Supplementary Material, adaptive weighting shows promising performance in terms of the percentages of correctly identified \onetoone and \onetomany mappings.  Theoretical justification of adaptive thresholding warrants future research.

	Improving semantic interoperability of EHR data is a pressing need for both clinical practice and biomedical research. Our proposed novel code translation method offers a scalable and automated approach for EHR data harmonization. A caveat is that the heterogeneity in medical coding could be partially driven by patient characteristics. In particular, the {within-group} coding differences mainly correspond to coding practice heterogeneity, whereas the {across-group} differences, if present, may reflect patient population heterogeneity. 
			However, we believe that the practice patterns of medicine,  the clinical knowledge, and the comorbidity patterns of diseases are mostly shared and hence transferable across healthcare systems.
			Thus, provided that both healthcare systems have sufficient number of patients with the diseases that the ICD codes cover, the embeddings trained from different healthcare systems still have the potential of being translated.
			In addition, the proposed method depends on the assumption that the SEVs can be aligned via a rotation $\Wbb$, and the block-diagonal structure of $\bPi$ is correctly specified. Model diagnosis, sensitivity analysis with more flexible models, and validation of the group structure are hence imperative. Nevertheless, the iSphereMAP algorithm remains meaningful as the cosine similarity measures the closeness of code pairs regardless of the adequacy of the vMF model assumption. 
			A potential limitation is the need of expert knowledge to further investigate whether any statistical finding corresponds to an actual mapping between two sets of medical codes.
			Another limitation is the lack of symmetry in the learned mapping, which is a key issue of the state-of-the-art language translation algorithms. Learning a symmetric translation is an open question very much of interest.

	\clearpage
	\begin{center}
	{\large\bf Supplementary Material for ``Spherical Regression under Mismatch Corruption with Application to Automated Knowledge Translation"}
	\end{center}
	\setcounter{section}{0}
	\setcounter{figure}{0}
	\renewcommand\thesection{\Alph{section}}
	\renewcommand\thesubsection{\thesection.\arabic{subsection}}
	\renewcommand\thefigure{\thesection.\arabic{figure}}
	\newtheorem{prop}{Proposition}
	\newtheorem{innercustomlem}{Lemma}
	\newenvironment{customlem}[1]
	{\renewcommand\theinnercustomlem{#1}\innercustomlem}
	{\endinnercustomlem}
	\numberwithin{equation}{section}
	\numberwithin{theorem}{section}
	\numberwithin{lemma}{section}
	\numberwithin{remark}{section}
	\numberwithin{corollary}{section}
	\numberwithin{prop}{section}
	\medskip
	
	\begin{spacing}{1.5}	
		\section{Word embedding algorithms} \label{supp:W2Valgorithm}
		Word embedding is the collective name for a set of language modeling and feature learning techniques in natural language processing. 
		Essentially, words can be represented as low-dimensional vectors of real numbers, often referred to as word representations or word embeddings, such that words with similar meanings will be closer to each other.
		The embeddings are often trained in an unsupervised manner, and can then be used as input features in supervised tasks.
		The idea of word representation stems from a psychological claim that human learn the meaning of a word from its context.
		Specifically, words with similar meanings will tend to occur in similar contexts, and thus co-occurrence of words carries key information for learning semantic representations. This idea has led to two main streams of word embedding algorithms: (1) context prediction based, which makes predictions of a neighbor word within local context windows using neural network, such as the \texttt{word2vec} \citep{mikolov2013distributed}; (2) co-occurrence count based, which explicitly factorizes a word-context matrix that measures mutual information based on the co-occurrence count, such as latent semantic analysis (LSA) \citep{deerwester1990indexing} and global vectors (GloVe) \citep{pennington2014glove}.
		
		To fix notation, let $w\in V$ denote a word and $c\in V$ denote its context within a pre-specified window, where $V$ is the vocabulary, i.e. the collection of all words that appear in a specific {corpus}. Let $D=D_{{corpus}}$ denote the collection of observed word-context pairs $(w,c)$ in the {corpus}. Let $\#(w,c)$  denote the number of times the pair $(w, c)$ appears in $D$. Therefore $|D|=\sum_{w,c\in V}\#(w,c)$. Note that it is possible that $\#(w,c)=0$ for a particular word-context pair $(w,c), w,c\in V$. Let $p$ denote a prespecified dimension of the embeddings with $p\leq |V|$; let $\vec{w}$ and $\vec{c}$ denote the $p$-dimensional embeddings of a word $w$ and a context $c$ respectively; and let $W_{|V|\times p}=[\vec{w}_1,\dots,\vec{w}_{|V|}]\trans$ and $C_{|V|\times p}=[\vec{c}_1,\dots,\vec{c}_{|V|}]\trans$ denote the matrix of all word and context embeddings respectively. Further, define point-wise mutual information (PMI) for a word-context pair $(w,c)$ as \[\text{PMI}(w,c)=\log\frac{\hat{P}(w,c)}{\hat{P}(w)\hat{P}(c)},\] where $\hat{P}(w,c)=\#(w,c)/|D|, \hat{P}(w)=\sum_{c'}\#(w,c')/|D|, \hat{P}(c)=\sum_{w'}\#(w',c)/|D| $. Because when $\#(w,c)=0$, we have $\log(\#(w,c))=-\infty$, we further introduce the positive point-wise mutual information (PPMI), which is $PPMI(w,c)=\max(\text{PMI}(w,c),0)$.

		The set of \texttt{word2vec} algorithms utilize a single-layer neural network for prediction with neural network architectures such as continuous bag-of-words (CBOW) or skip-gram (SG). These prediction-based training algorithms can be further combined with negative sampling (NS) which randomly selects a small number of “negative” words and update their embeddings. Recently it has been shown that the skip-gram combined with negative-sampling (SGNS) implicitly factorizes a shifted pointwise mutual information matrix \citep{levy2014neural}, i.e., \[W\cdot C\trans\approx M^{PMI}-\log k \mathbf{1}\mathbf{1}\trans,\] where $M^{PMI}$ is a $|V|\times |V|$ matrix with $M_{w,c}^{PMI}=\text{PMI}(w,c)$, and $k$ is the prespecified number of negative samples. This discovery connected the prediction-based and count-based language models and showed that the underlying statistics for both models is the co-occurrence count.
		
		In contrast to the implicit factorization in \texttt{word2vec}, the GloVe explicitly factorizes a log-count matrix shifted by word/context-specific bias terms, i.e., \[W\cdot C\trans\approx M^{\log(\#(w,c))}-\vec{b_{w}}\mathbf{1}\trans-\mathbf{1}\vec{b_{c}}\trans,\] where $M^{\log(\#(w,c))}$ is a $|V|\times |V|$ log-count matrix with $M_{w,c}^{\log(\#(w,c))}=\log(\#(w,c))$, $\vec{b_{w}}$ and $\vec{b_{c}}$ are unknown bias terms for the word and the context that are estimated in parallel with the embeddings $W$ and $C$.
		It has been shown that the different performances of different word embedding algorithms are largely due to system design choices and hyperparameter optimizations, rather than the embedding algorithms themselves \citep{levy2015improving}. There is no global advantage to any single approach over the others.
		
		In Table~\ref{w2vtable} we provide a summary of different word embedding algorithms. These algorithms either implicitly or explicitly factorize a matrix derived from the co-occurrence of words and contexts. In SGNS with $k$ negative samples, the corresponding shifted point-wise mutual information (SPMI) derived in \cite{levy2014neural} is given by $\text{PMI}(w,c)-\log k$, although in practice the shifted \textit{positive} point-wise mutual information (SPPMI) may be used instead of the SPMI, which is defined as $\text{SPPMI}(w,c)=\max(\text{PMI}(w,c)-\log k,0)$.

		\begin{table}[h]
			{\scriptsize{
					\centering
					\begin{tabular}{|c|c|c|}
						\hline 
						Method & Low rank approximation & Definitions of matrices derived from co-occurrence\tabularnewline
						\hline 
						\hline 
						Basic Semantic Vector & \multirow{2}{*}{$\vec{w}=M_{w,\cdot}^{PPMI}$} & $M_{w,c}^{PMI}=\text{PMI}(w,c)=\log\frac{\hat{P}(w,c)}{\hat{P}(w)\hat{P}(c)}$,\tabularnewline
						(no dimensional reduction) &  & where $\frac{\hat{P}(w,c)}{\hat{P}(w)\hat{P}(c)}=\frac{\#(w,c)\cdot|D|}{\sum_{c'}\#(w,c')\cdot\sum_{w'}\#(w',c)}$\tabularnewline
						\hline 
						Traditional singular & \multirow{2}{*}{$W\cdot C^{T}\approx M^{PPMI}$} & $M_{w,c}^{PPMI}=\text{PPMI}(w,c)=\max(\text{PMI}(w,c),0)$\tabularnewline
						value decomposition (SVD) &  & $\vec{w}\cdot\vec{c}\approx\text{PMI}(w,c)$\tabularnewline
						\hline 
						Skip-Grams with & \multirow{2}{*}{$W\cdot C^{T}\approx M^{SPPMI}-\log k$} & $M_{w,c}^{SPPMI}=\text{SPPMI}(w,c)=\max(\text{PMI}(w,c)-\log k,0)$\tabularnewline
						Negative Sampling (SGNS) &  & $\vec{w}\cdot\vec{c}\approx\text{PMI}(w,c)-\log k$\tabularnewline
						\hline 
						\multirow{2}{*}{Global Vectors (GloVe)} & \multirow{2}{*}{$W\cdot C^{T}\approx M^{\log(\#(w,c))}-\vec{b_{w}}-\vec{b_{c}}$} & $\vec{w}\cdot\vec{c}+b_{w}+b_{c}\approx\log(\#(w,c))$,\tabularnewline
						&  & where $b_{w}$ and $b_{c}$ (scalars) are word/context-specific biases\tabularnewline
						\hline 
					\end{tabular}
			}}
			\caption{Summary of state-of-the-art word embedding algorithms}\label{w2vtable}
		\end{table}

		In our data application described in Section~\ref{application}, we define co-occurrence of a pair of ICD codes as the number of patients who has  the pair of ICD codes co-occur in his/her health record within any 30-day period. We count the co-occurrence of all pairs of ICD-9 codes within each healthcare system in order to translate codes between two healthcare systems as described in Section~\ref{example1}. That is, we obtain two co-occurrence matrices from the two healthcare systems. From each co-occurrence matrix we derive the SPPMI matrix.
		Similarly, we count the co-occurrence of all pairs of ICD-9 codes and the co-occurrence of all pairs of ICD-10 codes using the EHR data of the Veterans Health Administration (VHA) in order to map from ICD-9 codes to ICD-10 codes as described in Section~\ref{example2}. Again we derive two SPPMI matrices from the two sets of co-occurrence counts.
		
		We use singular value decomposition (SVD) of the SPPMI matrix to generate semantic embedding vectors (SEVs). Specifically, we first compute the SPPMI matrix $M^{\text{SPPMI}}$ with each element defined as $M^{\text{SPPMI}}_{w,c}=\max(\log \frac{\#(w,c)}{\sum_{c'}\#(w,c')\cdot (\sum_{w'}\#(w',c)/|D|)^{\alpha}}-\log k,0)$. Here $\alpha$ is a smoothing parameter that aims to shrink the PMI of any $w$ co-occurring with a rare context $c$, which was shown to be an effective way to tune the PMI \citep{levy2015improving}. Then, for a pre-specified dimension $p$, we compute the rank $p$ approximation of $M^{\text{SPPMI}}$.		
		That is, we obtain  $M^{\text{SPPMI}}\approx U_p\Sigma_p U_p\trans$, where $U_p$ is the first $p$ eigenvectors of $M^{\text{SPPMI}}$, and $\Sigma_p$ is a diagonal matrix of the first $p$ eigenvalues of $M^{\text{SPPMI}}$. 
		Finally, we obtain $W=U_p\sqrt{\Sigma_p}$. If $M^{\text{SPPMI}}$ is not symmetric, then we have $M^{\text{SPPMI}}\approx U_p\Sigma_p V_p\trans$ and we obtain $W=U_p\sqrt{\Sigma_p}+V_p\sqrt{\Sigma_p}$.
		We set the number of negative samples $k=10$, and smoothing parameter $\alpha=0.75$. The SVD of the SPPMI matrix was implemented using the augmented implicitly restarted Lanczos bidiagonalization algorithm  \citep{baglama2005augmented} with the \texttt{irlba} package in R \citep{baglama2017irlba}.

		\section{Proof of supporting lemmas}\label{app-lemma}
		
		\begin{customlem}{\ref{lemma:eta-bound}}
			For $p\geq 4$ and $\kappa>0$, $\max\{0,1-\frac{p-1}{2\kappa}\}< \gamma_{\kappa,p}< 1$.
		\end{customlem}
		
		\begin{proof}[Proof]
			Without loss of generality, assume $\bmu=(1,0,...,0)$. Then,
			$\gamma_{\kappa,p}=\Ebb(Z_1),$
			where $\bZ=(Z_1,..,Z_p)\trans \sim \text{vMF}_{\bmu,\kappa,p}$.
			The moment generating function of $Z_1$ as
			$M_{Z_1}(\lambda)=C_p(\kappa)/C_p(\kappa+\lambda)$ as shown in the proof of Proposition~\ref{lemma:tail-bound-von-miss}.
			Thus, we have
			\begin{equation*}
			\gamma_{\kappa,p}=\Ebb(Z_1)=(\log M_{Z_1}(\lambda))'|_{\lambda=0}
			= -\frac{C_p'(\kappa)}{C_p(\kappa)}
			= \frac{B'_{p/2-1}(\kappa)}{B_{p/2-1}(\kappa)}-\frac{p/2-1}{\kappa}.
			\end{equation*}
			According to the equation below (2.6) in \cite{baricz2010bounds}, we have
			\begin{equation*}
			\frac{B'_{p/2-1}(\kappa)}{B_{p/2-1}(\kappa)}\kappa> \kappa-1/2,
			\end{equation*}
			for $p\geq 4$.
			Combining the above two inequalities, we have
			\begin{equation*}
			\gamma_{\kappa,p}\geq 1-\frac{1}{2\kappa}- \frac{p-2}{2\kappa}=\max\{ 1- \frac{p-1}{2\kappa},0\}.
			\end{equation*}
			
		\end{proof}
		
		\begin{lemma}\label{lemma:variance-von-miss}
			If $\bZ\sim \text{vMF}_{\bmu,\kappa,p}$ and $\kappa>0$, then
			$E[\|\bZ-\gamma_{\kappa,p}\bmu\|^2] = 1-\gamma_{\kappa,p}^2$.
		\end{lemma}
		
		\begin{proof}[Proof]
			Let $\bZ=(Z_1,...,Z_p)\trans $ defined similarly as above. We have,
			\begin{equation*}\begin{split}
			E[\|\bZ-\gamma_{\kappa,p}\bmu\|^2]
			=& \Ebb((Z_1-\gamma_{\kappa,p})^2+Z_2+...+Z_p^2 )=  \Ebb((Z_1-\gamma_{\kappa,p})^2)+\Ebb(1-Z_1^2)\\
			=&  1 - 2\gamma_{\kappa,p}\Ebb(Z_1)+\gamma_{\kappa,p}^2
			=  1- \gamma_{\kappa,p}^2.
			\end{split}\end{equation*}
		\end{proof}
		
		\begin{lemma}\label{lemma:vec-norm-cos}
			For a vector $\bZ=(Z_1,...,Z_p)\in \mathcal{R}^p$, if $\|\bZ-\Ibb_{j\cdot}\|_2\leq r$ for $0<r<\frac{1}{2}$, then
			\begin{equation*}
			\cos(\bZ,\Ibb_{j\cdot})\geq 1-2r.
			\end{equation*}
		\end{lemma}
		
		\begin{proof}[Proof]
			First, from $\|\bZ-\Ibb_{j\cdot}\|_2\leq r$, we have
			$1-r\leq\|\bZ\|_2\leq1+r.$
			Since $\|\bZ-\Ibb_{j\cdot}\|_2^2= \|\bZ\|_2^2+1-2Z_j$, we have
			\begin{equation*}
			|Z_j-1|=\frac{1}{2}|\|\bZ\|_2^2-1-\|\bZ-\Ibb_{j\cdot}\|_2^2|
			\leq\frac{1}{2}\{r(2+r)+ r^2\}=r(1+r).
			\end{equation*}
			It follows that
			\begin{equation*}
			\cos(\bZ,\Ibb_{j\cdot})=\frac{Z_j}{\|\bZ\|_2}
			\geq \frac{1-r(1+r)}{1+r}
			=1-\frac{r(2+r)}{1+r}\geq1-2r.
			\end{equation*}
		\end{proof}
		
		\begin{lemma}\label{lemma:cos-lower}
			For three vectors $\bX,\bY,\bZ\in \mathcal{R}^p$. If
			$\cos(\bX,\bY)\geq 1-\alpha$ and $\cos(\bY,\bZ)\leq \beta$, then $\cos(\bX,\bZ)\leq \beta+\sqrt{2\alpha}$.
		\end{lemma}
		\begin{proof}[Proof]
			Without loss of generality, assume $\|\bX\|_2=\|\bY\|_2=\|\bZ\|_2=1$ and $\bY=(1,0,\dots,0)$. Then, we know $X_1\geq 1-\alpha$ and $Z_1\leq \beta$. Now we consider $\cos(\bX,\bZ)$. We have
			\begin{equation*}
			\begin{split}
			\cos(\bX,\bZ)=\sum_{i=1}^p X_iZ_i
			\leq X_1Z_1 + (\sum_{i=2}^p X_i^2)^{1/2}(\sum_{i=2}^p Z_i^2)^{1/2}
			= X_1Z_1+(1-X_1^2)^{1/2}(1-Z_1)^{1/2}.
			\end{split}
			\end{equation*}
			By assumptions on $X_1$ and $Z_1$, we further have
			\begin{equation*}
			X_1Z_1+(1-X_1^2)^{1/2}(1-Z_1)^{1/2}
			\leq Z_1 + (1-X_1^2)^{1/2}
			\leq \beta+(1-(1-\alpha)^2)^{1/2}
			\leq \beta + \sqrt{2\alpha}.
			\end{equation*}
			Combining the above two displays, we completes the proof.
		\end{proof}
		
		\begin{lemma}\label{remark:pi}
			To guarantee that $\bPi\Xbb$ is still on the hypersphere, $\bPi$ has to satisfy the following inequality
			\begin{equation*}
			\frac{1}{\sqrt{n_k}}\leq \|\bPi_{i\cdot}\|_2\leq \frac{1}{\sigma_{n_k}(\Xbb_{[G_k,:]})},\text{ for all $i\in G_k$}.
			\end{equation*}		
		\end{lemma}
		\begin{proof}[Proof]
			The spherical requirement is $1=\|\bPi\addstar_{i\cdot}\Xbb \|_2=\|\bPi\addstar_{[i,G_k]}\Xbb _{[G_k,:]}\|_2.$		On the other hand, we know 
			\begin{equation*}
			\|\bPi\addstar_{[i,G_k]}\|_2 \sigma_{n_k}(\Xbb _{[G_k,:]}) \leq\|\bPi\addstar_{[i,G_k]}\Xbb _{[G_k,:]}\|_2\leq \sigma_1(\Xbb _{[G_k,:]})\|\bPi\addstar_{[i,G_k]}\|_2\leq \sqrt{n_k}\|\bPi\addstar_{[i,G_k]}\|_2.
			\end{equation*}
			Thus, 
			\begin{equation*}
			\frac{1}{n_k}\leq \|\bPi\addstar_{[i,G_k]}\|_2
			\leq \frac{1}{\sigma_{n_k}(\Xbb _{[G_k,:]})}
			\end{equation*}
			
		\end{proof}

		\section{Tail analysis of the vMF distribution} \label{app-tail}
		
		\begin{prop}\label{lemma:tail-bound-von-miss}
			Let $\bmu\in \hypersphere$, $\bZ\sim \text{vMF}_{\bmu,\kappa,p}$, and $\bepsilon = \bZ-\bmu$. Then, for $p\geq 4$ and $\frac{p-1}{2\kappa}\leq \delta\leq 2$, the following statements hold.
			\begin{enumerate}
				\item $P(\bepsilon\trans \bmu\leq -\delta)\leq \exp\{
				- \delta\kappa  + \frac{1}{2}(p-1)(\log \kappa+1) - \frac{1}{2}(p-1)\log (\frac{\frac{1}{2}(p-1)}{\delta})
				\}$;
				\item $P(\|\bepsilon\|_2\geq \sqrt{2\delta})\leq \exp\{
				- \delta\kappa  + \frac{1}{2}(p-1)(\log \kappa+1) - \frac{1}{2}(p-1)\log (\frac{\frac{1}{2}(p-1)}{\delta})
				\}$.
				\item If we have $Q_1,...,Q_m$ be i.i.d copies of $\|\bepsilon\|_2^2$, then for $s\geq 0$, 
				\begin{equation}\label{eq:sum-tail}
				P\left\{\sum_{i=1}^m Q_i\geq \frac{m(p-1)}{\kappa} + \frac{m(p-1)}{\kappa}s \right\}
				\leq \exp\left\{
				- \frac{m(p-1)}{2}(s-\log (1+s))\right\}.
				\end{equation}
				\item Let $\{Q_{k,l},k=1, ..., K, l = 1, ..., n_k\}$ be $n = \sum_{k=1}^K n_k$ i.i.d realizations of
				$\|\bepsilon\|_2^2$. Then, for each $t>0$,
				\begin{equation}\label{eq:max-sum-tail}
				P\left\{\max_{1\leq k\leq K }\sum_{l=1}^{n_k} Q_{k,l}\geq 
				\frac{n_{\max}(p-1)}{\kappa}(1+s_t)\right\}\leq e^{-t},
				\end{equation}
				where $n_{\max} = \max_{1 \le k \le K} n_k$, $s_t\geq 0$ is the unique solution to 
				$s_t-\log(1+s_t)=\{2(\log K+t)\}/\{(p-1)n_{\min}\}$
				and $n_{\min} = \min_{1 \le k \le K} n_k$. 
				{ In particular, if $4\log K\le (p-1)n_{\min}$ and $t=\log K$, then $s_t\leq 3$ and}
				\begin{equation}
				P\left(\max_{1\leq k\leq K }\sum_{l=1}^{n_k} Q_{k,l}\geq 
				\frac{4n_{\max}(p-1)}{\kappa}\right)\leq \frac{1}{K}.
				\end{equation}
			\end{enumerate}

		\end{prop}
		
		\begin{remark}
			The second tail bound implies that $ \|\bepsilon\|_2^2 =O_p( p/\kappa )$ and $\bepsilon\trans \bmu = O_p( p/\kappa)$. 
		\end{remark}
		\begin{remark}
			For $1\leq p\leq 3$, less sharp tail bounds can also be developed.
		\end{remark}
		
		\begin{proof}[Proof of Proposition~\ref{lemma:tail-bound-von-miss}]
			Without loss of generality, we assume $\bmu=(1,0,...,0)$. Then, $\bepsilon = (Z_1 - 1,0,\dots,0)$, and $P(\bepsilon\trans \bmu\leq -\delta) = P( 1-Z_1\geq \delta )$.
			Using the Chernoff bound \citep{chernoff1952measure}, we can see that for all $\lambda>0$,
			\begin{equation}\label{eq:upper-projected-von-miss}
			P\left( 1-Z_1\geq \delta\right)
			\leq e^{-\lambda \delta}e^{\lambda} \Ebb(e^{-\lambda Z_1}) = e^{\lambda(1-\delta)}\Ebb(e^{-\lambda Z_1}).
			\end{equation}
			We proceed to calculate the moment generating function $\Ebb(e^{-\lambda Z_1})$.
			Let $f_{Z_1}(z_1)$ be the density function of $Z_1$.
			According to the density function of $\bZ$, we have the marginal density,
			\begin{equation}
			f_{Z_1}(z_1)= C_p(\kappa)\exp(\kappa z_1) \omega_{p-2}\left(\sqrt{1-z_1^2}\right),
			\end{equation}
			where $\omega_d(r)$ denotes the surface area of a $d-1$-dimensional sphere (living in a $d$-dimensional space) with the radius $r$, and $C_p(\kappa)= \kappa^{p/2-1}/\{(2\pi)^{p/2}B_{p/2-1}(\kappa)\}$ is the normalizing constant for vMF distribution, and $B_{\nu}(x)$ denotes the modified Bessel function. Then, 
			\begin{equation}
			\begin{split}
			\Ebb(e^{-\lambda Z_1})
			=& \int_{-1}^1 e^{-\lambda z_1} C_p(\kappa)\exp(\kappa z_1) \omega_{p-2}(\sqrt{1-z_1^2}) d z_1\\
			= & \frac{C_p(\kappa)}{C_p(\kappa-\lambda)} =
			\left(\frac{\kappa}{\kappa-\lambda}\right)^{\nu}\frac{B_{\nu}(\kappa-\lambda)}{B_{\nu}(\kappa)},
			\end{split}
			\end{equation}
			where we let $\nu=\frac{p}{2}-1$.
			Combining this with \eqref{eq:upper-projected-von-miss}, we have
			\begin{equation}\label{eq:one-step-tail-project}
			P( 1-Z_1\geq \delta)
			\leq e^{-\lambda \delta}e^{\lambda}\left(\frac{\kappa}{\kappa-\lambda}\right)^{\nu}\frac{B_{\nu}(\kappa-\lambda)}{B_{\nu}(\kappa)}.
			\end{equation}
			We use the following upper bound of $\frac{B_{\nu}(\kappa-\lambda)}{B_{\nu}(\kappa)}$, which is the equation (2.6) in \cite{baricz2010bounds}. For all $\nu\geq \frac{1}{2}$ and $0<x<y$,
			\begin{equation*}
			\frac{B_{\nu}(x)}{B_{\nu}(y)}
			<e^{x-y}\left(\frac{y}{x}\right)^{1/2}.
			\end{equation*}
			Setting $x=\kappa-\lambda$ and $y=\kappa$ in the above display and combining it with \eqref{eq:one-step-tail-project}, we have 
			\begin{equation*}
			P(1-Z_1\geq \delta)
			\leq \inf_{0\leq \lambda\leq \kappa}e^{-\lambda \delta}\left(\frac{\kappa}{\kappa-\lambda}\right)^{\nu+\frac{1}{2}}.
			\end{equation*}
			If $\kappa-\frac{\nu+\frac{1}{2}}{\delta}\geq 0$, 
			\begin{equation*}
			\inf_{0\leq \lambda\leq \kappa}e^{-\lambda \delta}\left(\frac{\kappa}{\kappa-\lambda}\right)^{\nu+\frac{1}{2}}
			= \exp\left\{
			- \delta\kappa  + (\nu+\frac{1}{2})(\log \kappa+1) - (\nu+\frac{1}{2})\log (\frac{\nu+\frac{1}{2}}{\delta})
			\right\},
			\end{equation*}
			where the minimum is achieved at $\lambda=\kappa-\frac{\nu+\frac{1}{2}}{\delta}$.
			Summarizing the above results, we have
			\begin{equation}\label{eq:bound-one-error}
			P(\bepsilon\trans \bmu\leq -\delta)\leq \exp\left\{
			- \delta\kappa  + \frac{1}{2}(p-1)(\log \kappa+1) - \frac{1}{2}(p-1)\log (\frac{\frac{1}{2}(p-1)}{\delta})
			\right\}
			\end{equation}
			for $p\geq 4$ and $\delta\geq\frac{p-1}{2\kappa}$.
			The tail bound of $\|\bepsilon\|_2$ is  straightforward based on the above inequality, because $\|\bepsilon\|_2^2 = 2(1-\bmu\trans \bZ)$.
			
			To establish \eqref{eq:sum-tail}, we note that from a similar Chernoff bound, 
			\begin{equation*}
			P(\sum_{i=1}^m Q_i\geq 2m\delta)
			\leq \inf_{\lambda\geq 0}\left(e^{-\lambda \delta}(\frac{\kappa}{\kappa-\lambda})^{\nu+\frac{1}{2}}\right)^m.
			\end{equation*}
			for $\delta\geq\frac{p-1}{2\kappa}$. According to \eqref{eq:bound-one-error}, the above display is simplified as
			\begin{equation*}
			P\left(\sum_{i=1}^m Q_i\geq 2m\delta\right)
			\leq \exp\left\{ m\left(
			- \delta\kappa  + \frac{1}{2}(p-1)(\log \kappa+1) - \frac{1}{2}(p-1)\log (\frac{\frac{1}{2}(p-1)}{\delta})\right)
			\right\}.
			\end{equation*}
			Let $\delta=\frac{p-1}{2\kappa}(1+s)$ in the above display for $s\geq 0$ and simplifying it, we arrive at
			\begin{equation}
			P\left(\sum_{i=1}^m Q_i\geq 2m\frac{p-1}{2\kappa}(1+s)\right)
			\leq \exp\left\{ -\frac{m(p-1)}{2}(s-\log(1+s)).
			\right\}.
			\end{equation}

			For \eqref{eq:max-sum-tail}, we first observe that for each $k$, $1\leq k\leq K$, according to \eqref{eq:sum-tail}, we have
			\begin{equation*}
			P\left(\sum_{l=1}^{n_k}Q_{k,l}\geq\frac{n_k(p-1)}{\kappa}(1+s_t)\right)
			\leq\exp\left\{-\frac{n_k}{n_{\min}}(\log K+t)\right\}
			\leq e^{-t}/K.
			\end{equation*}
			This further gives
			\begin{equation}
			\label{vmt1}
			\begin{split}
			P\left(\sum_{l=1}^{n_k}Q_{k,l}\geq\frac{n_{\max}(p-1)}{\kappa}(1+s_t)\right)
			\leq& K^{-1}e^{-t}.
			\end{split}
			\end{equation}
			By the union bound, we have
			\begin{equation*}
			P\left\{\max_{1\leq k\leq K}\sum_{l=1}^{n_k}Q_{k,l}\geq\frac{n_{\max}(p-1)}{\kappa}(1+s_t)\right\}
			\leq \sum_{i=1}^K P\left\{\sum_{j=1}^{n_k}Q_{k,l}\geq\frac{n_{\max}(p-1)}{\kappa}(1+s_t)\right\}
			\leq e^{-t},
			\end{equation*}
			where the last inequality is obtained by \eqref{vmt1}.
			This completes the proof for \eqref{eq:max-sum-tail}.
			
			If $4\log K \le (p-1)n_{\min}$ and $t=\log K$, then 
			$s_t-\log(1+s_t)\leq 1$. It follows that $s_t<3$ and
			\begin{equation*}
			P\left\{\max_{1\leq k\leq K}\sum_{j=1}^{n_k}Q_{k,l}\geq\frac{4n_{\max}(p-1)}{\kappa}\right\}
			\leq 	P\left\{\max_{1\leq k\leq K}\sum_{j=1}^{n_k}Q_{k,l}\geq\frac{n_{\max}(p-1)(1+s_t)}{\kappa}\right\}\\
			\leq e^{-t} = \frac{1}{K}.
			\end{equation*}
			
		\end{proof}
		
		\section{The iterative spherical regression mapping (iSphereMAP) algorithm}\label{supp:ispheremap}
		Below we detail our proposed iterative spherical regression mapping algorithm. Although the iSphereMAP procedure can iterate until convergence, we find that the estimators stabilize after three steps. 
		\begin{algorithm}[!htp]
			\small
			\caption{The iSphereMAP algorithm.}
			\textbf{Input}  spherical data $\Xbb$ and $\Ybb$ ordered by group, group information $G_k, k=1,\dots,K$ that defines data within a group $\Xbb_{[G_k,:]}$ and $\Ybb_{[G_k,:]}$, and the block-diagonal structure in the mapping matrix $\bPi=\text{diag}\{\bPi^1,\dots,\bPi^K\}$, tuning parameter $\lambda_n$ selected by cross-validation optimizing the mean squared error for prediction of $\Ybb$.		\vspace{3mm}\\
			\noindent\textbf{Initialize} the mapping matrix $\bPihat\supone=\Ibb$.\vspace{3mm}\\
			\noindent \textbf{Three-step procedure} where steps 2 and 3 can be iterated
			\begin{enumerate}
				\item[] \textbf{Step 1: spherical regression} to estimate the orthogonal translation matrix as
				$$\Wbbhat\supone = \argmin{\Wbb: \Wbb\Wbb\trans  = \Ibb_p}{\| \Ybb- \Xbb \Wbb\|_F^2}.$$
				\item[] \textbf{Step 2: map data} to obtain the mapping matrix $\bPihat\suptwo$ by the following two substeps
				\begin{enumerate}
					\item[(1)] \textbf{Ordinary least squares} to estimate an initial mapping matrix $\Initialpi = \diag\{\Initialpi^1, ..., \Initialpi^K\}$, where for each block $$\Initialpi^k=
					\Ybb_{[G_k,:]} (\Xbb_{[G_k,:]}\Wbbhat\supone)\trans (\Xbb_{[G_k,:]}\Xbb_{[G_k,:]}\trans )^{-1}.$$
					\item[(2)]	\textbf{Hard-thresholding} as follows $$\bPihat\suptwo_{i\cdot}=   \Ibb_{\jtilde_i\cdot} \mathbbm{1}(\widetilde{\beta}_i\leq\lambda_n) + 
					\frac{\Initialpi_{i\cdot}}{\|(\Initialpi_{i\cdot}\Xbb)\trans\|_2} \mathbbm{1}(\widetilde{\beta}_i > \lambda_n),$$ where $\widetilde{\beta}_i=1-\max_{j:j\sim i}\; \cos(\Initialpi_{i\cdot},\Ibb_{j\cdot})$ measures the distance between $\Initialpi_{i\cdot}$ and a \onetoone mapping $\Ibb_{j\cdot}$.
				\end{enumerate}
				\item[] \textbf{Step 3: refined spherical regression} using matched data to update the translation matrix as
				$$\Wbbhat\suptwo = \argmin{\Wbb: \Wbb\Wbb\trans  = \Ibb_p}{\| \Ybb_{\ii}- \Xbb_{\ii} \Wbb\|_F^2},$$ where $\mm(\bPihat\suptwo)=\{i \in \Gsc: \bPihat\suptwo_{i\cdot}=\Ibb_{i\cdot} \}$ indexes the set of matched units as determined by $\bPihat\suptwo$.
			\end{enumerate}
			\textbf{Output} Mapping matrix $\bPihat\suptwo$, orthogonal translation matrix $\Wbbhat\suptwo$.
		\end{algorithm}

		\section{Proof of theorems and corollaries} \label{app-theorem}
		\subsection{Proof of Theorem \ref{thm:w-hat-consistent}}\label{app-theorem-one}
		\begin{proof}[Proof of Theorem~\ref{thm:w-hat-consistent}]
			Write $\bepsilon_i = \bY_i - \gamma_{\kappa,p}\Wbb\trans(\bPi_{i\cdot}  \Xbb)\trans$ and 
			$\Vbb=\Ybb -E(\Ybb )=(\bepsilon_1, ..., \bepsilon_n)\trans$, where
			\begin{equation*}
			\gamma_{\kappa,p}= \frac{B'_{p/2-1}(\kappa)}{2B_{p/2-1}(\kappa)}-\frac{p/2-1}{\kappa}.
			\end{equation*}
			We have
			\begin{equation}
			\Ybb =\gamma_{\kappa,p}\bPi\addstar  \Xbb   \Wbb \addstar +\Vbb
			= \gamma_{\kappa,p}\Xbb   \Wbb \addstar  + \gamma_{\kappa,p}(\bPi\addstar -\Ibb)\Xbb   \Wbb \addstar  + \Vbb.
			\end{equation}
			Recall that we write $\Usc(A)=A(A\trans A)^{-1/2}$ for the polar decomposition of $A$. Then, by definition, 
			\begin{equation}\label{eq:w-hat-polar}
			\Wbbhat\supone=\Usc(\Xbb\trans \Ybb )= \Usc(\gamma_{\kappa,p} \Xbb\trans \Xbb  \Wbb \addstar + \Delta), 
			\quad \mbox{where} \ \Delta =\gamma_{\kappa,p}\Xbb\trans (\bPi\addstar -\Ibb)\Xbb   \Wbb \addstar  + \Xbb\trans \Vbb.
			\end{equation}
			On the other hand, since $\Xbb\trans \Xbb  $ is positive definite with smallest eigenvalue $\sigma_p(\Xbb  )>0$, 
			\begin{equation}\label{eq:w-polar}
			\Usc(\gamma_{\kappa,p}\Xbb\trans \Xbb  \Wbb \addstar) = \Xbb\trans \Xbb  \Wbb \addstar \Wbb \addstartop(\Xbb\trans \Xbb  )^{-1}\Wbb \addstar  = \Wbb \addstar .
			\end{equation}
			\eqref{eq:w-hat-polar} and \eqref{eq:w-polar} together imply 
			\begin{equation}
			\Wbbhat\supone-\Wbb \addstar = \Usc(  \gamma_{\kappa,p}\Xbb\trans \Xbb  \Wbb \addstar + \Delta) - \Usc(\gamma_{\kappa,p} \Xbb\trans \Xbb  \Wbb \addstar).
			\end{equation}
			We proceed to obtain an upper bound on $\|\Usc(\gamma_{\kappa,p} \Xbb\trans \Xbb  \Wbb \addstar + \Delta) - \Usc( \gamma_{\kappa,p}\Xbb\trans \Xbb  \Wbb \addstar)\|$, 
			where $\|\cdot\|$ denotes a unitary invariant matrix norm. We use results in Lemma \ref{lemma-mathias}, which is a slight modification of  Theorem 2.4 in \cite{mathias1993perturbation}. 
			\begin{lemma}[Modification of Theorem 2.4 in \cite{mathias1993perturbation}] \label{lemma-mathias}
				Let $A,\Delta$ be two $p\times p$ real matrices. Assume that $\sigma_p(A)-\sigma_1(\Delta)>0$. Then, for any unitary invariant norm $\|\cdot\|$,
				\begin{equation*}
				\|\Usc(A+\Delta)-\Usc(A)\|\leq 2
				[\sigma_p(A) + \sigma_{p-1}(A)-2\sigma_1(\Delta)]^{-1}\|\Delta\|.
				\end{equation*}
			\end{lemma}
			Let $A=\gamma_{\kappa,p}\Xbb\trans \Xbb  \Wbb \addstar$ in Lemma \ref{lemma-mathias} . For any unitary invariant norm $\|\cdot\|$, we have
			\begin{equation}\label{eq:perturb-polor}
			\|\Usc( \gamma_{\kappa,p}\Xbb\trans \Xbb  \Wbb \addstar + \Delta) - \Usc( \gamma_{\kappa,p}\Xbb\trans \Xbb  \Wbb \addstar)\|
			\leq (\sigma_p(A) - \sigma_1(\Delta))^{-1}\| \Delta\|.
			\end{equation}
			To bound the right-hand side of the above display, we note that
			\begin{align}
			\sigma_p(A) & \geq\gamma_{\kappa,p}\sigma_p(\Xbb  )^2\sigma_p(\Wbb \addstar )
			=\gamma_{\kappa,p}\sigma_p(\Xbb  )^2, 	\label{eq:A-lower} \\
			\mbox{and}\quad \sigma_1(\Delta) & \leq \|\Delta\|_F\leq \gamma_{\kappa,p}
			\|\Xbb\trans (\bPi\addstar -\Ibb)\Xbb   \Wbb \addstar \|_F +\|\Xbb\trans \Vbb\|_F
			\label{eq:e1-bound-main}
			\end{align}
			Recall that $\mathcal{D}\equiv \Dsc(\bPi,\Ibb)=\{
			i\in\Gsc: \bPi\addstar_{i\cdot}\neq \Ibb_{i\cdot}
			\}$ indexes the mismatched rows.
			Then,
			\begin{equation}\label{eq:e1-bound-first}
			\begin{split}
			& \|\Xbb\trans (\bPi\addstar -\Ibb)\Xbb   \|_F
			= \|(\Xbb  _{[\mathcal{D},:]})\trans (\bPi\addstar_{[\mathcal{D},:]}-\Ibb_{[\mathcal{D},:]})\Xbb  \|_F
			\leq  \|\Xbb  _{[\mathcal{D},:]}\|_F\|(\bPi\addstar_{[\mathcal{D},:]}-\Ibb_{[\mathcal{D},:]})\Xbb  \|_F\\
			\leq & \|\Xbb  _{[\mathcal{D},:]}\|_F\{\|\bPi\addstar_{[\mathcal{D},:]}\Xbb  \|_F+\|\Ibb_{[\mathcal{D},:]}\Xbb  \|_F\}
			=2 \nmis .
			\end{split}
			\end{equation}
			For the last line of the above display, we used the spherical assumption and obtain that $\|\Xbb  _{[\mathcal{D},:]}\|_F=\sqrt{\nmis}$ and $\|\bPi\addstar_{[\mathcal{D},:]}\Xbb  \|_F=\sqrt{\nmis}$.
			Combining \eqref{eq:A-lower}, \eqref{eq:perturb-polor}, and \eqref{eq:e1-bound-first}, we have
			\begin{equation*}
			\|\Usc(  \Xbb\trans \Xbb  \Wbb \addstar + \Delta) - \Usc(  \Xbb\trans \Xbb  \Wbb \addstar)\|
			\leq \{\gamma_{\kappa,p}\sigma_p(\Xbb  )^2 - 2\gamma_{\kappa,p}\nmis- \|\Xbb\trans \Vbb\|_F\}^{-1}\|\Delta\|.
			\end{equation*}
			That is, 
			\begin{equation*}
			\|\Wbbhat\supone-\Wbb \addstar \|
			\leq  \{\gamma_{\kappa,p}\sigma_p(\Xbb  )^2 - 2\gamma_{\kappa,p}\nmis- \|\Xbb\trans \Vbb\|_F\}^{-1}\| \Delta\|.
			\end{equation*}
			In particular, if we take $\|\cdot\|$ to be $\|\cdot\|_F$ in the above inequality, then
			\begin{equation}\label{eq:stepone-w-error}
			\|\Wbbhat\supone-\Wbb \addstar \|_F
			\leq  \{\gamma_{\kappa,p}\sigma_p(\Xbb  )^2 - 2\gamma_{\kappa,p}\nmis- \|\Xbb\trans \Vbb\|_F\}^{-1}(2\gamma_{\kappa,p}\nmis+ \|\Xbb\trans \Vbb\|_F).
			\end{equation}
			To analyze the tail behavior of $\|\Xbb\trans \Vbb\|_F$, we note that
			\begin{equation}\label{eq:xe-f-norm-square}
			\begin{split}
			\Ebb(\|\Xbb\trans \Vbb\|_F^2)
			=  \Ebb\left(
			\tr(\Vbb\trans \Xbb  \Xbb\trans \Vbb)
			\right)= \Ebb\left(
			\tr(\Xbb  \Xbb\trans \Vbb\Vbb\trans )
			\right) =\tr\left(
			\Xbb  \Xbb \trans  \Ebb[\Vbb\Vbb\trans ]
			\right).
			\end{split}
			\end{equation}
			Since $\Vbb=(\bepsilon_1,...,\bepsilon_n)\trans $ and $\bepsilon_i$'s are centered and independent random vectors,	we have
			\begin{equation}
			\Ebb[\Vbb\Vbb\trans ]
			= (\Ebb\bepsilon_i\trans \bepsilon_j)_{1\leq i,j\leq n}
			=\mbox{diag}(\Ebb(\|\bepsilon_1\|_2^2,...,\Ebb(\|\bepsilon_n\|_2^2)).
			\end{equation}
			From Lemma \ref{lemma:variance-von-miss} in Appendix \ref{app-lemma}, the distribution of $\|\bepsilon_i\|_2^2$ does not depend on $\bmu$ and 
			\begin{equation}\label{eq:eet-exp}
			\Ebb[\Vbb\Vbb\trans ]
			= (\Ebb\bepsilon_i\trans \bepsilon_j)_{1\leq i,j\leq n}
			=(1-\gamma_{\kappa,p}^2)\Ibb_n.
			\end{equation}
			On the other hand,  the diagonal elements of $\Xbb \Xbb \trans $ are all ones because of the spherical assumption. Combining this fact with \eqref{eq:xe-f-norm-square} and \eqref{eq:eet-exp}, we arrive at
			\begin{equation}\label{corollaryuse35}
			\Ebb(\|\Xbb \trans \Vbb\|_F^2)
			= (1-\gamma_{\kappa,p}^2)\tr\left(\Xbb \Xbb \trans \right)
			= n (1-\gamma_{\kappa,p}^2).
			\end{equation}
			Now we apply Chebyshev inequality to $\|\Xbb \trans \Vbb\|_F$ and obtain that for all $t>0$
			\begin{equation}\label{eq:xe-cheby}
			P(\|\Xbb \trans \Vbb\|_F\geq t)
			\leq t^{-2} n (1-\gamma_{\kappa,p}^2),
			\end{equation}
			or, equivalently,
			\begin{equation}
			P(\|\Xbb \trans \Vbb\|_F\geq t\sqrt{n (1-\gamma_{\kappa,p}^2)})
			\leq t^{-2}
			\end{equation}
			for all $t>0$.
			Combining \eqref{eq:xe-cheby} and \eqref{eq:stepone-w-error}, we arrive at
			\begin{equation}\label{eq:what-bound-with-gamma}
			\|\Wbbhat\supone-\Wbb \addstar \|_F
			\leq \frac{2\gamma_{\kappa,p} \nmis+ t \sqrt{n(1-\gamma_{\kappa,p}^2)}}{\gamma_{\kappa,p}\sigma_p(\Xbb )^2 - 2\gamma_{\kappa,p} \nmis- t \sqrt{n(1-\gamma_{\kappa,p}^2)}}
			\end{equation}
			with probability that is at least $1-1/t^2$.
			
		\end{proof}
		
		\subsection{Proof of Corollary~\ref{coro:w-hat-large-kappa}}\label{app-coro-one}
		\begin{proof}[Proof of Corollary~\ref{coro:w-hat-large-kappa}]
			The proof for the case where both $p$ and $\kappa$ are fixed is straightforward because $\boundgamma<\gamma_{\kappa,p}<1$ is a constant. 
			Now we consider the case when $\kappa\to\infty$ and $p\geq4$. 
			By the assumption that $\gamma_{\kappa,p}>\boundgamma$ is bounded away from zero, we have
			\[
			\|\Wbbhat\supone-\Wbb \addstar \|_F
			\leq \frac{2 \nmis+ t \sqrt{n\eta_{\kappa,p}/\gamma_{\kappa,p}}}{\sigma_p(\Xbb )^2 - 2 \nmis- t \sqrt{n\eta_{\kappa,p}/\gamma_{\kappa,p}}}
			\] with probability at least $1-1/t^2$. Therefore, we have
			\[
			\|\Wbbhat\supone-\Wbb\|_F=O_P\left(\frac{\sqrt{n\eta_{\kappa,p}}+\nmis}{\sigma_p(\Xbb)^{2}} \right)
			\]
			if $\sqrt{n\eta_{\kappa,p}}=o( \sigma_p(\Xbb)^2)$ and $\nmis=o( \sigma_p(\Xbb)^2)$.  In particular, when $\nmis=o(\sigma_p(\Xbb))$ and $\sqrt{np/\kappa}=o(\sigma_p(\Xbb))$, we have
			$\|\Wbbhat\supone-\Wbb\|_F=o_p(1)$.
			
		\end{proof}
		
		\subsection{Proof of Theorem~\ref{thm:thm-OLS-pi}}\label{app-theorem-two}
		\begin{proof}[Proof of Theorem~\ref{thm:thm-OLS-pi}]
			For any $k \in \{1, ...,K\}$, let $\Ubb_{[G_k,:]}= \Ybb_{[G_k,:]}- \bPi\addstar_{[G_k,G_k]} \Xbb_{[G_k,:]}  \Wbb \addstar $ be the $n_k \times p$ residual matrix. 
			The OLS estimator for $\bPi^k$ is
			\begin{equation*}
			\Initialpi^k \equiv \Initialpi\trans _{[G_k,G_k]}
			=  (\Xbb _{[G_k,:]}\Xbb _{[G_k,:]}\trans )^{-1}\Xbb _{[G_k,:]}\Wbbhat\supone\Ybb _{[G_k,:]}\trans =
			\Ybb _{[G_k,:]}{(\Wbbhat\supone)}\trans \Xbb _{[G_k,:]}\trans (\Xbb _{[G_k,:]}\Xbb _{[G_k,:]}\trans )^{-1}.
			\end{equation*}
			Let $\Delta \Wbb =\Wbbhat\supone-\Wbb \addstar $, then
			\begin{equation*}
			\begin{split}
			\Initialpi^k
			=& 
			(\bPi^k\addstar  \Xbb _{[G_k,:]} \Wbb \addstar  + \Ubb_{[G_k,:]})
			(\Wbb \addstartop+\Delta \Wbb \trans )\Xbb _{[G_k,:]}\trans (\Xbb _{[G_k,:]}\Xbb _{[G_k,:]}\trans )^{-1}\\
			= &\bPi^k\addstar 
			+ \bPi^k\addstar  \Xbb _{[G_k,:]} \Wbb \addstar (\Delta \Wbb )\trans \Xbb _{[G_k,:]}\trans (\Xbb _{[G_k,:]}\Xbb _{[G_k,:]}\trans )^{-1}\\
			& + \Ubb_{[G_k,:]}\Wbb \addstartop\Xbb _{[G_k,:]}\trans (\Xbb _{[G_k,:]}\Xbb _{[G_k,:]}\trans )^{-1}
			+ \Ubb_{[G_k,:]}(\Delta \Wbb )\trans \Xbb _{[G_k,:]}\trans (\Xbb _{[G_k,:]}\Xbb _{[G_k,:]}\trans )^{-1}.
			\end{split}
			\end{equation*}
			In what follows, we find an upper bound of
			\begin{equation*}
			\begin{split}
			&\|\bPi^k\addstar  \Xbb _{[G_k,:]} \Wbb \addstar (\Delta \Wbb )\trans \Xbb _{[G_k,:]}\trans (\Xbb _{[G_k,:]}\Xbb _{[G_k,:]}\trans )^{-1}\|_F\\
			& + \|\Ubb_{[G_k,:]}\Wbb \addstartop\Xbb _{[G_k,:]}\trans (\Xbb _{[G_k,:]}\Xbb _{[G_k,:]}\trans )^{-1}\|_F
			+ \|\Ubb_{[G_k,:]}(\Delta \Wbb )\trans \Xbb _{[G_k,:]}\trans (\Xbb _{[G_k,:]}\Xbb _{[G_k,:]}\trans )^{-1}\|_F.
			\end{split}
			\end{equation*}
			Fro the first term, we note that $\sigma_1\{\Xbb _{[G_k,:]}\trans (\Xbb _{[G_k,:]}\Xbb _{[G_k,:]}\trans )^{-1}\} = \sigma_1 \{(\Xbb _{[G_k,:]}\trans \Xbb _{[G_k,:]}\trans )^{-1}\}^{1/2}
			=\sigma_{n_k}(\Xbb _{[G_k,:]})^{-1}$ 
			and hence
			\begin{equation*}
			\begin{split}
			\|\bPi^k\addstar  \Xbb _{[G_k,:]} \Wbb \addstar (\Delta \Wbb )\trans \Xbb _{[G_k,:]}\trans (\Xbb_{[G_k,:]}\Xbb _{[G_k,:]}\trans )^{-1}\|_F
			\leq & 
			\|\bPi^k\addstar  \Xbb _{[G_k,:]}\|_2\|\Delta \Wbb \|_F
			\|\Xbb _{[G_k,:]}(\Xbb _{[G_k,:]}\Xbb _{[G_k,:]}\trans )^{-1}\|_2\\
			\leq & \sigma_{n_k}(\Xbb _{[G_k,:]})^{-1}\sigma_1(\bPi^k\addstar  \Xbb _{[G_k,:]})\|\Delta \Wbb \|_F,
			\end{split}
			\end{equation*}
			where for a matrix $\bA$, $\|\bA\|_2$ denotes its spectral norm.
			For the second term, we have
			\begin{equation*}
			\|\Ubb_{[G_k,:]}\Wbb \addstartop\Xbb _{[G_k,:]}\trans (\Xbb _{[G_k,:]}\Xbb _{[G_k,:]}\trans )^{-1}\|_F\\	
			\leq  \sigma_{n_k}(\Xbb _{[G_k,:]})^{-1}\|\Ubb_{[G_k,:]}\|_F.		
			\end{equation*}
			For the third term, we have
			\begin{equation*}
			\begin{split}
			\|\Ubb_{[G_k,:]}(\Delta \Wbb )\trans \Xbb _{[G_k,:]}\trans (\Xbb _{[G_k,:]}\Xbb _{[G_k,:]}\trans )^{-1}\|_F
			\leq \sigma_{n_k}(\Xbb _{[G_k,:]})^{-1}\|\Ubb_{[G_k,:]}\|_F\|\Delta \Wbb \|_F.
			\end{split}
			\end{equation*}
			Combining these inequality, we have
			\begin{equation*}
			\begin{split}
			\|\Initialpi^k-\bPi\addstar_{[G_k,G_k]}\|_F
			\leq \sigma_{n_k}(\Xbb _{[G_k,:]})^{-1}\{
			\|\Ubb_{[G_k,:]}\|_F(1+\|\Delta \Wbb \|_F) + \sigma_1(\bPi^k\addstar  \Xbb _{[G_k,:]})\|\Delta \Wbb \|_F
			\}.
			\end{split}
			\end{equation*}
			We combine our analysis for different and arrive at		
			\begin{equation*}
			\begin{split}
			&\max_{1\leq k\leq K}\|\Initialpi^k-\bPi\addstar^k\|_F\\
			\leq & [\min_{1\leq k\leq K}\sigma_{n_k}(\Xbb _{[G_k,:]})]^{-1}\{
			\max_{1\leq k\leq K}\|\Ubb_{[G_k,:]}\|_F(1+\|\Delta \Wbb \|_F) + \max_{1\leq k\leq K}\sigma_1(\bPi^k\addstar  \Xbb _{[G_k,:]})\|\Delta \Wbb \|_F
			\}\\
			= & [\min_{1\leq k\leq K}\sigma_{n_k}(\Xbb _{[G_k,:]})]^{-1}\{
			(1+\|\Delta \Wbb \|_F)\max_{1\leq k\leq K} \|\Ubb_{[G_k,:]}\|_F+\|\Delta \Wbb \|_F\max_{1\leq k\leq K}\sqrt{n_k}
			\}.
			\end{split}
			\end{equation*}
			We proceed to analyzing the probabilistic properties of the above display. From Corollary~\ref{coro:w-hat-large-kappa}, we know that under the assumptions of Corollary~\ref{coro:w-hat-large-kappa}, 
			\begin{equation*}
			\|\Delta \Wbb \|_F=O_p(\sigma_p(\Xbb )^{-2}(\sqrt{n\eta_{\kappa,p}}+\nmis) )=o_p(1).
			\end{equation*}

			For $\max_{1\leq k\leq K} \|\Ubb_{[G_k,:]}\|_F$, we apply \eqref{eq:max-sum-tail} in Proposition~\ref{lemma:tail-bound-von-miss}. Then, we have that with probability at least $1-\frac{1}{K}$, 
			\begin{equation*}
			\max_{1\leq k\leq K} \|\Ubb_{[G_k,:]}\|_F
			\leq 2\max_{1\leq k\leq K}\sqrt{n_k}(\frac{p-1}{\kappa})^{1/2},
			\end{equation*}
			given that $4\log K\leq (p-1)n_{\min}$.
			Combining these, we have with the probability going to one, 
			\begin{equation*}
			\max_{1\leq k\leq K}\|\Initialpi^k-\bPi\addstar^k\|_F
			\leq 4c_n , \mbox{where}\ 
			c_n = \frac{\max_{1\leq k\leq K}\sqrt{n_k}\left\{
				\left(\frac{p}{\kappa}\right)^{1/2}+\sigma_p(\Xbb )^{-2}\left(\sqrt{n\eta_{\kappa,p}}+\nmis\right)
				\right\}}{\min_{1\leq k\leq K}\sigma_{n_k}(\Xbb _{[G_k,:]})}.
			\end{equation*}
			
			Assuming that $p=o(\kappa)$, we further have that with the probability converging to one, 
			\begin{equation*}
			\max_{1\leq k\leq K}\|\Initialpi^k-\bPi\addstar^k\|_F	\leq 
			4c_n.
			\end{equation*}
			which implies that $\|\Initialpi-\bPi\addstar \|_2= O_p\left(c_n\right)$.
			
		\end{proof}
		
		\subsection{Proof of Theorem~\ref{thm:thresholding}}\label{app-theorem-three}
		\begin{proof}[Proof of Theorem~\ref{thm:thresholding}]
			From Theorems~\ref{thm:w-hat-consistent} and \ref{thm:thm-OLS-pi},  we have for any $a_n\to\infty$, 
			$P(F_n)\to 1,$ where
			\begin{equation}
			F_n=\Big\{\max_{1\leq k\leq K}\|\Initialpi_k-\bPi\addstar_k\|_F
			\leq a_n c_n \text{ and }\|\Delta \Wbb \|_F\leq a_n\sigma_p(\Xbb )^{-2}(\sqrt{n\eta_{\kappa,p}}+\nmis)
			\Big\},
			\end{equation}
			and $c_n$ is defined above.
			From now on, we restrict our analysis on the event $F_n$ with some suitable choice of $a_n$. We first observe that when $F_n$ occurs, for each row
			\begin{equation}\label{eq:row-bound}
			\|\Initialpi_{i\cdot}-\bPi\addstar_{i\cdot}\|_2
			\leq d_n,
			\end{equation}
			where $d_n=a_n c_n$.
			We first use Lemma \ref{lemma:vec-norm-cos} to show that, 
			if $\bPi\addstar_{i\cdot}=\Ibb_{j\cdot}$ for some $j$, then $\bPihat\suptwo_{i\cdot}=\Ibb_{j\cdot}$. 
			In other words, we show that for all $i \notin\mathcal{C},\; \bPihat\suptwo_{i\cdot}=\bPi_{i\cdot}=\Ibb_{j\cdot}$.
			Lemma  \ref{lemma:vec-norm-cos} and \eqref{eq:row-bound} imply that if 
			$2d_n<\lambda_n<\frac{1}{2}$
			and $\bPi_{i\cdot}\addstar =\Ibb_{j\cdot}$ for some $j$, then we get $\bPihat\suptwo_{i\cdot}=\Ibb_{j\cdot}$. This result holds for all rows $i\notin \mathcal{C}$. Thus, given
			$c_n\ll\lambda_n<\frac{1}{2},$	we have the exact recovery for rows $i\notin\mathcal{C}$ on the event $F_n$ with any sequence $a_n$ such that $a_n\to\infty$ and $a_n\ll\lambda_n/c_n$.
			
			It remains to show that the hard thresholding does not have any effect on the rows with $i\in\mathcal{C}$. 
			We note that
			\begin{equation*}
			\left\|\frac{\Initialpi_{i\cdot}}{\|\bPi\addstar_{i\cdot}\|_2}-\frac{\bPi\addstar_{i\cdot}}{\|\bPi\addstar_{i\cdot}\|_2}\right\|_2\leq\frac{d_n}{\|\bPi_i\addstar \|_2}.
			\end{equation*}
			Similar to Lemma~\ref{lemma:vec-norm-cos}, we have
			\begin{equation}
			\cos\left(\frac{\Initialpi_{i\cdot}}{\|\bPi\addstar_{i\cdot}\|_2}, \frac{\bPi\addstar_{i\cdot}}{\|\bPi\addstar_{i\cdot}\|_2}\right)
			\geq 1-\frac{2d_n}{\|\bPi_i\addstar \|_2}.
			\end{equation}
			To bound $\cos(\Initialpi_{i\cdot},\Ibb_{j\cdot})$, we use Lemma \ref{lemma:cos-lower} in Appendix \ref{app-lemma} by setting $\bX=\Initialpi_{i\cdot}$, $\bY=\bPi\addstar_{i\cdot}$ and $\bZ=\Ibb_{j\cdot}$. 
			It follows that
			\begin{equation*}
			\cos(\Initialpi_{i\cdot},\Ibb_{j\cdot})
			\leq 1-\min_{i\in \mathcal{C} } \beta_j+2\sqrt{\frac{d_n}{\|\bPi_{i\cdot}\addstar \|_2}} = 1-\Bsc_{\min}+2\sqrt{\frac{d_n}{\|\bPi_{i\cdot}\addstar \|_2}}.
			\end{equation*}
			From Lemma~\ref{remark:pi}, $\|\bPi_{i\cdot}\addstar \|_2\geq 1/\sqrt{n_k}$. Thus, we arrive at
			\begin{equation*}
			\cos(\Initialpi_{i\cdot},\Ibb_{j\cdot})
			\leq 1-			\Bsc_{\min}+2\sqrt{d_n\max_{1\leq k\leq K}\sqrt{n_k}},
			\end{equation*}
			which implies
			\begin{equation*}
			1-	\widetilde{\beta}_i=	\max_{j:j\sim i}\cos(\Initialpi_{i\cdot},\Ibb_{j\cdot})
			\leq 1-			\Bsc_{\min}+2\sqrt{d_n\max_{1\leq k\leq K}\sqrt{n_k}}.
			\end{equation*}
			That is,
			$\widetilde{\beta}_i\geq \Bsc_{\min}-2\sqrt{d_n\max_{1\leq k\leq K}\sqrt{n_k}}.$
			Because we do hard-thresholding only when $	\widetilde{\beta}_i
			\leq \lambda_n$, and from the theorem assumptions we have
			$
			\Bsc_{\min}-2\sqrt{d_n\max_{1\leq k\leq K}\sqrt{n_k}}>\lambda_n
			$
			for large $n$, we can see that the hard-thresholding will not have any effect to the $i$'th row of $\Initialpi$ for sufficiently large $n$. 			This completes our proof for the model selection consistency part.

			We proceed to the estimation error bound of $\bPihat\suptwo_{i\cdot}$ for $i\in \mathcal{C}$. Without loss of generality, suppose $i\in G_k$.
			Recall that $\bPihat\suptwo_{i\cdot}=\xi_i\Initialpi_{i\cdot},$			where $\xi_i=\|\Initialpi_{i\cdot}\Xbb \Wbbhat\supone\|_2^{-1}$.
			Clearly,
			\begin{equation}\label{eq:bound-non-comb}
			\|\bPihat\suptwo_{i\cdot}-\bPi\addstar_{i\cdot}\|_2
			\leq \|\Initialpi_{i\cdot}\|_2 |\xi_i-1| + \|\Initialpi_{i\cdot}-\bPi\addstar_{i\cdot}\|_2
			\leq (\|\bPi\addstar_{i\cdot}\|_2+d_n)|\xi_i-1|+d_n.
			\end{equation}
			Now we consider an upper bound on $|\xi_i-1|$.
			We observe that for $i\in G_k$,
			\begin{equation}
			\begin{split}
			|1-1/\xi_i|=& \left|\|\Initialpi_{i\cdot}\Xbb \Wbbhat\supone\|_2-1\right|
			\leq \|\Initialpi_{i\cdot}\Xbb \Wbbhat\supone- \bPi\addstar_{i\cdot}\Xbb \Wbb \addstar \|_2\\
			\leq & \|\Initialpi_{i\cdot}\Xbb \Wbbhat\supone- \bPi\addstar_{i\cdot}\Xbb \Wbbhat\supone\|_2+\|\bPi\addstar_{i\cdot}\Xbb \Wbbhat\supone- \bPi\addstar_{i\cdot}\Xbb \Wbb \addstar \|_2\\
			\leq& d_n\sigma_1(\Xbb _{[G_k,:]})+\sigma_1(\bPi\addstar_{i\cdot}\Xbb )\|\Delta \Wbb \|_F
			\leq d_n\sqrt{n_k}+ \|\Delta \Wbb \|_F\\
			\leq & d_n \sqrt{n_k}+ a_n\sigma_p(\Xbb )^{-2}(\sqrt{n\eta_{\kappa,p}}+\nmis)
			\leq  2 d_n \sqrt{n_k}.
			\end{split}
			\end{equation}
			It follows that $|\xi_i-1|\leq|1-1/\xi_i|/(1/\xi_i)
			\leq 2 d_n \sqrt{n_k}/(1-2 d_n \sqrt{n_k}).$			Under  assumption of the theorem, for $n$ sufficiently large,  $d_n \sqrt{n_k}<1/4$,  we have
			$|\xi_i-1|\leq 4d_n \sqrt{n_k}<1.$	Combining this inequality with \eqref{eq:bound-non-comb} and the fact that $d_n\sqrt{n_k}<\frac{1}{4}$ again, we have
			\begin{equation}\label{eq:est-bound-mismatch}
			\|\bPihat\suptwo_{i\cdot}-\bPi\addstar_{i\cdot}\|_2
			\leq d_n\{(\|\bPi\addstar_{i\cdot}\|_2+d_n)4 \sqrt{n_k}+1\}
			\leq d_n (4\|\bPi\addstar_{i\cdot}\|_2\sqrt{n_k}+2)
			\leq 6 d_n \|\bPi\addstar_{i\cdot}\|_2\sqrt{n_k}.
			\end{equation}
			\sloppy To get the last inequality in the above display, we used Lemma~\ref{remark:pi}.
			In particular, if $
			c_n \max_{1\leq k\leq K}\max_{i\in\mathcal{C},i\in G_k}\|\bPi\addstar_{i\cdot}\|_2\sqrt{n_k}\to 0,$
			then with $a_n$ chosen such that $a_n\to\infty$ and $a_nc_n \max_{1\leq k\leq K}\max_{i\in\mathcal{C},i\in G_k}\|\bPi\addstar_{i\cdot}\|_2\sqrt{n_k}\to 0$, we have
			$d_n \max_{1\leq k\leq K}\max_{i\in\mathcal{C},i\in G_k}\|\bPi\addstar_{i\cdot}\|_2\sqrt{n_k}\to 0,$
			where $d_n=a_nc_n$.
			This together with \eqref{eq:est-bound-mismatch} implies that 
			\begin{equation*}
			\max_{i\in\mathcal{C}}\|\bPihat\suptwo_{i\cdot}-\bPi\addstar_{i\cdot}\|_2\to 0
			\end{equation*}
			on the event $F_n$. That is, all rows of $\bPihat\suptwo_{i\cdot}$ are consistent when $i\in\mathcal{C}$.
			
		\end{proof}

		\subsection{Proof of Corollary~\ref{coro:refinement}}\label{app-coro-two}	
		\begin{proof}[Proof of Corollary~\ref{coro:refinement}]
			The subsample we use to obtain $\Wbbhat\suptwo$ includes $\Xbb\ii$ and 
			\[
			\Ybb\ii=\gamma_{\kappa,p}\Xbb\ii\Wbb+\gamma_{\kappa,p}(\bPi\ii -\Ibb\ii)\Xbb \Wbb+\Vbb\ii,
			\]
			with sample size  $|\mm(\widehat{\bPi}\suptwo)|$.
			Therefore, we have that the refined estimate
			\[
			\begin{aligned}
			\widehat{\Wbb}\suptwo=&\Usc(\Xbb\ii\trans \Ybb\ii)\\
			=&\Usc(\gamma_{\kappa,p}\Xbb\ii\trans \Xbb\ii\Wbb+\gamma_{\kappa,p}\Xbb\ii\trans (\bPi\ii -\Ibb\ii)\Xbb\Wbb+\Xbb\ii\trans \Vbb\ii)
			\end{aligned}
			\]
			Let $A=\gamma_{\kappa,p}\Xbb\ii\trans \Xbb\ii\Wbb$, and 
			$$\Delta=\gamma_{\kappa,p}\Xbb\ii\trans (\bPi\ii -\Ibb\ii)\Xbb\Wbb+\Xbb\ii\trans \Vbb\ii.$$ 
			Then by the same arguments as the proof of Theorem~\ref{thm:w-hat-consistent}, we have $\Usc(A)=\Wbb$, and $\sigma_p(A)\geq \gamma_{\kappa,p}\sigma_p(\Xbb\ii)^2$.
			In addition, by Lemma 3 we have
			\[\|\Usc(A+\Delta)-\Usc(A)\|_F\leq (\sigma_p(A)-\sigma_1(\Delta))^{-1}\|\Delta\|_F\leq (\sigma_p(A)-\|\Delta\|_F)^{-1}\|\Delta\|_F.\]
			For $\|\Delta\|_F$, by the same argument as the proof of Theorem~\ref{thm:w-hat-consistent}, we have 
			\begin{equation}\label{deltaFana}
			\begin{aligned}
			\|\Delta\|_F\leq& \gamma_{\kappa,p}\|\Xbb\ii\trans  (\bPi\ii -\Ibb\ii)\Xbb\Wbb\|_F+\|\Xbb\ii\trans \Vbb\ii\|_F\\
			\leq& 2\gamma_{\kappa,p}|\mathcal{D}(\bPi\ii,\Ibb\ii)|+\|\Xbb\ii\trans \Vbb\ii\|_F
			\end{aligned}
			\end{equation}
			Next, define the event $\mathcal{I}=\{\mathcal{D}(\bPi\ii,\Ibb\ii)=\emptyset\}$, then for a positive $t$, we have $P(\|\Delta\|_F> t)\leq P(\mathcal{I}^c)+P(\|\Delta\|_F> t,\mathcal{I})$. First, under the assumptions in Theorem~\ref{thm:thresholding}, $P(\mathcal{I}^c)\to0$. Second, by (\ref{deltaFana}) we have
			\[
			\begin{aligned}
			P(\|\Delta\|_F> t,\mathcal{I})\leq& P(2\gamma_{\kappa,p}|\mathcal{D}(\bPi\ii,\Ibb\ii)|+\|\Xbb\ii\trans \Vbb\ii\|_F> t,\mathcal{I})\\
			=&P(\|\Xbb\ii\trans \Vbb\ii\|_F> t,\mathcal{I})
			=P(\|\Xbb_{[\mm(\bPi),:]}\trans \Vbb_{[\mm(\bPi),:]}\|_F> t,\mathcal{I})\\
			\leq&P(\|\Xbb_{[\mm(\bPi),:]}\trans \Vbb_{[\mm(\bPi),:]}\|_F> t).
			\end{aligned}
			\]
			By the Chebyshev inequality, we have
			\[
			\begin{aligned}
			P(\|\Delta\|_F> t,\mathcal{I})\leq& \frac{1}{t^2}\Ebb[\|\Xbb_{[\mm(\bPi),:]}\trans \Vbb_{[\mm(\bPi),:]}\|^2_F]
			= \frac{1}{t^2}(n-\nmis) \eta_{\kappa,p},
			\end{aligned}
			\]
			where the last equation follows the same argument as (\ref{corollaryuse35}), except the sample size here is $|\mm(\bPi)| = n - \nmis$ rather than $n$, with $\eta_{\kappa,p}=1-\gamma_{\kappa,p}^2$.
			It follows that 
			\begin{equation}\label{ppp}
			P(\|\Delta\|_F> t,\mathcal{I})\leq \frac{1}{t^2}(n-\nmis) \eta_{\kappa,p}
			\end{equation}  
			Therefore, 
			\begin{equation}\label{eq:original_refineWbound}
			P\left(\|\widehat{\Wbb}\suptwo-\Wbb \addstar \|_F
			\geq \frac{t \sqrt{(n-\nmis)\eta_{\kappa,p}}}{\gamma_{\kappa,p}\sigma_p(\Xbb_{[\mm(\bPi),:]} )^2 - t \sqrt{(n-\nmis)\eta_{\kappa,p}}}\right)
			\leq P(\mathcal{I}^c)+1/t^2,
			\end{equation}
			which further implies that as $n$ grows,
			\begin{equation*}
			\|\widehat{\Wbb}\suptwo-\Wbb \addstar \|_F=O_p\left\{\frac{ \sqrt{(n-\nmis)\eta_{\kappa,p}}}{\gamma_{\kappa,p}\sigma_p(\Xbb_{[\mm(\bPi),:]} )^2 }\right\} .
			\end{equation*}
			Note that \eqref{ppp} holds when assumptions of Theorem~\ref{thm:thresholding} are satisfied, under which we have
			\begin{equation}\label{eq:simplify_gamma_eta}
			\gamma_{\kappa,p}=1+o(1)\text{ and }
			\eta_{\kappa,p}=O(p/\kappa)=o(1).
			\end{equation}
			
			Moreover, by Weyl’s perturbation theorem (see, e.g. \cite{stewart1990computer}) and the fact that $\|\bX_i\|=1,\forall i$, we have $\sigma_p(\Xbb)^{2}- \nmis\leq \sigma_p(\Xbb_{[\mm({\bPi}),:]})^{2}\leq \sigma_p(\Xbb)^{2}$. By the assumption of Theorem~\ref{thm:thm-OLS-pi} that $\nmis=o(\sigma_p(\Xbb)^{2})$, we know that 
			\begin{equation}\label{eq:sigma_equiv}
			\sigma_p(\Xbb_{[\mm({\bPi}),:]})^{2}=(1+o(1))\sigma_p(\Xbb)^{2}.
			\end{equation}
			Because 
			$\sqrt{(n-\nmis)\eta_{\kappa,p}}<\sqrt{n\eta_{\kappa,p}}=o(\sigma_p(\Xbb)^2)$
			due to assumptions in Theorem~\ref{thm:thm-OLS-pi}, by (\ref{eq:sigma_equiv}) we have
			\begin{equation}\label{eq:spkapa}
			\sqrt{(n-\nmis)\eta_{\kappa,p}}=o(\sigma_p(\Xbb_{[\mm({\bPi}),:]})^2).
			\end{equation}
			
			Combining \eqref{eq:original_refineWbound}, \eqref{eq:simplify_gamma_eta}, \eqref{eq:sigma_equiv}, and \eqref{eq:spkapa}, we obtain that the error rate is
			\begin{equation}
			\|\widehat{\Wbb}\suptwo-\Wbb \addstar \|_F=O_p\left\{\frac{ \sqrt{(n-\nmis)\eta_{\kappa,p}}}{\sigma_p(\Xbb)^2}\right\}.
			\end{equation}
			
		\end{proof}
	\end{spacing}

	\newpage
	\section{Additional simulation results}\label{supp-additional-simu}
	We conduct additional simulation studies to evaluate the performance of our method under a few interesting scenarios in Sections~\ref{wrong_grp_info}-\ref{supp:lessnoise}, as well as to investigate the performance of an extension of the iSphereMAP algorithm in Section~\ref{supp:add_mismatch_1to1}.
	The data are generated following the same procedure as Section~\ref{simu} of the main manuscript except for the below specified distinctions.
	\subsection{Overly coarse group structure}\label{wrong_grp_info} 
	The block-diagonal structure of $\bPi$ is defined based on the group information, which may be inaccurate. In particular, if the block size is too small, then we misspecified $\bPi$ with too many zero entries. An extreme case is to specify that group size equals one, i.e. assuming $\bPi=\Ibb$. With such overly fine grouping, our estimator can miss a portion of mismatch patterns.
	
	In contrast, if the block size is too big, then we have a conservative model which could influence efficiency but not validity. In this section, we evaluate the performance of our proposed method in the scenario where the block diagonal structure in $\bPi$ is overly coarse. 
	Specifically, we generate the data following the procedure in Section~\ref{simu}. However, in the estimation procedure, the block diagonal structure in $\bPi$ is specified to be overly coarse, by combining two distinct groups into a larger group for every two out of five groups. 

	Following Section~\ref{simu}, we summarize in Figure~\ref{supp:fig:wrong_grp_info} the performance of $\bPihat\suptwo$ estimated from our method which utilizes incorrectly specified group information and from the MT method which does not utilize any group information, in terms of the match rate for \onetoone mapping and the mean squared error (MSE) of \onetomany mapping weight.
	We consider estimation via iSphereMAP under correct (in black) and overly coarse (in red) group information.
	We observe slightly increased MSE of the iSphereMAP estimator under overly coarse group structure. In addition, with sufficient sample size, overly coarse group structure has little impact on the \onetoone match rate. We thus generally recommend to be conservative in choosing the group structure to avoid model misspecification. In addition, our proposed method still outperforms the MT method which does not leverage group information. 
	
	\begin{figure}[!h]
		\begin{centering}
			\makebox[\textwidth][c]{\includegraphics[width=1\textwidth,
				scale=.5]{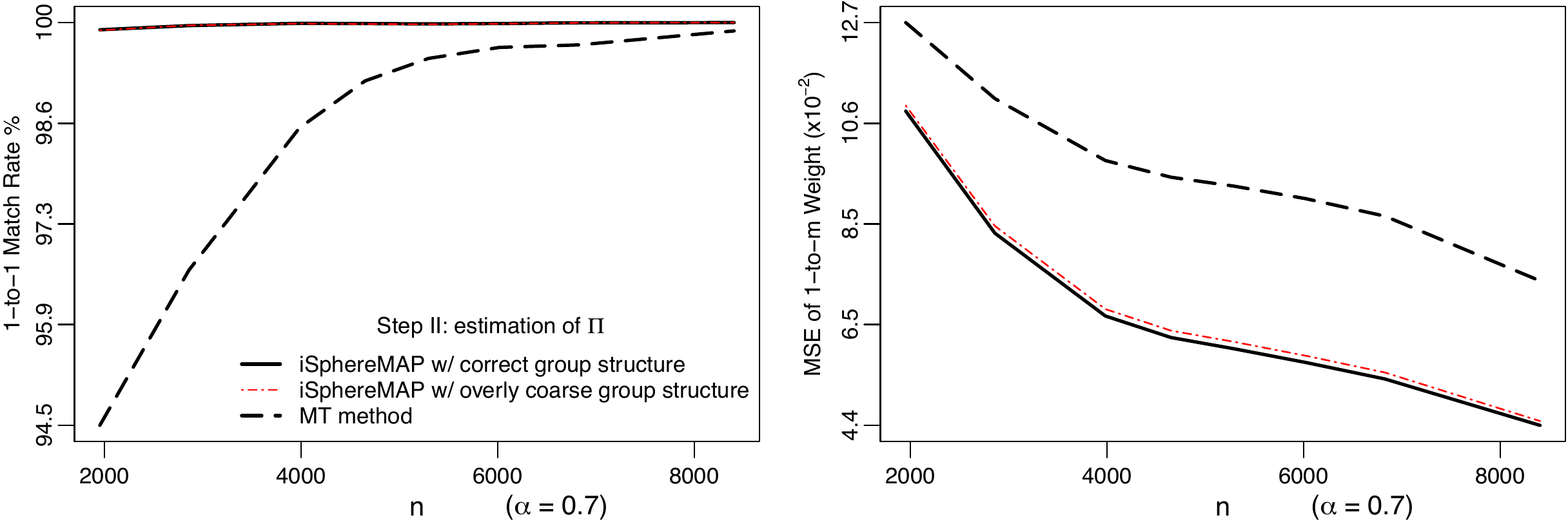}}
			\par \end{centering}
		\caption{\label{supp:fig:wrong_grp_info} Performance of $\bPihat\suptwo$ estimated from iSphereMAP with correct and overly coarse group structure as well as from the MT method in terms of the \onetoone match rate (left panel) and the MSE of \onetomany weight (right panel).
		}
	\end{figure}

	\subsection{Permutation only: no \onetomany mapping}\label{supp:perm_only_T}
	\begin{figure}[!h]
		\begin{centering}
			\makebox[\textwidth][c]{\includegraphics[height=5.5cm, scale=.4]{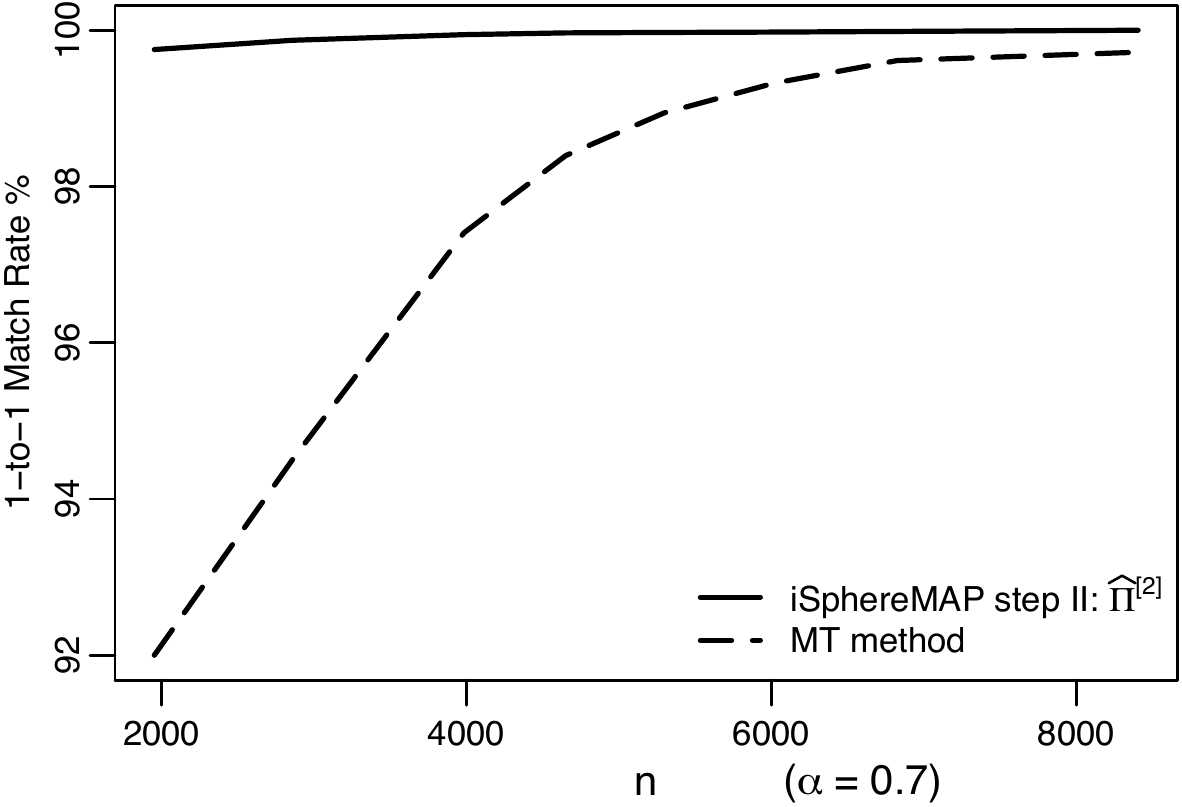}}
			\par \end{centering}
		\caption{\label{supp:fig:perm_only_T} Performance of $\bPihat\suptwo$ estimated from iSphereMAP and the MT method in terms of the \onetoone match rate in the scenario where the true $\bPi$ contains no \onetomany  mismatch but only \onetoone  mismatch.}
	\end{figure}
	We evaluate the match rate of our proposed method in the scenario where $\bPi$ is a permutation matrix. That is, only \onetoone  mapping is present, and there is no \onetomany mapping.
	Figure~\ref{supp:fig:perm_only_T} presents the performance of $\bPihat\suptwo$ estimated from our method with group information and from the MT method without group information, in terms of the match rate for \onetoone mapping.
	As is shown in Figure~\ref{supp:fig:perm_only_T}, the iSphereMAP estimator still outperforms the MT method. This is expected because the MT method aligns the SEV spaces via the ordinary least squares, which does not acknowledge the fact that all SEVs are unit-length vectors. In addition, it does not utilize the group information.

\subsection{Less noisy scenario}\label{supp:lessnoise}
We investigate the performance of our proposed method in the scenario where $\kappa=3000$. This is considered as a setting with less noise in data compared to the simulation studies in Section~\ref{simu} of the main manuscript. 

\begin{figure}[!h]
\begin{subfigure}[t]{1\textwidth}
\makebox[\textwidth][c]{\includegraphics[width=1\textwidth, scale=.5]{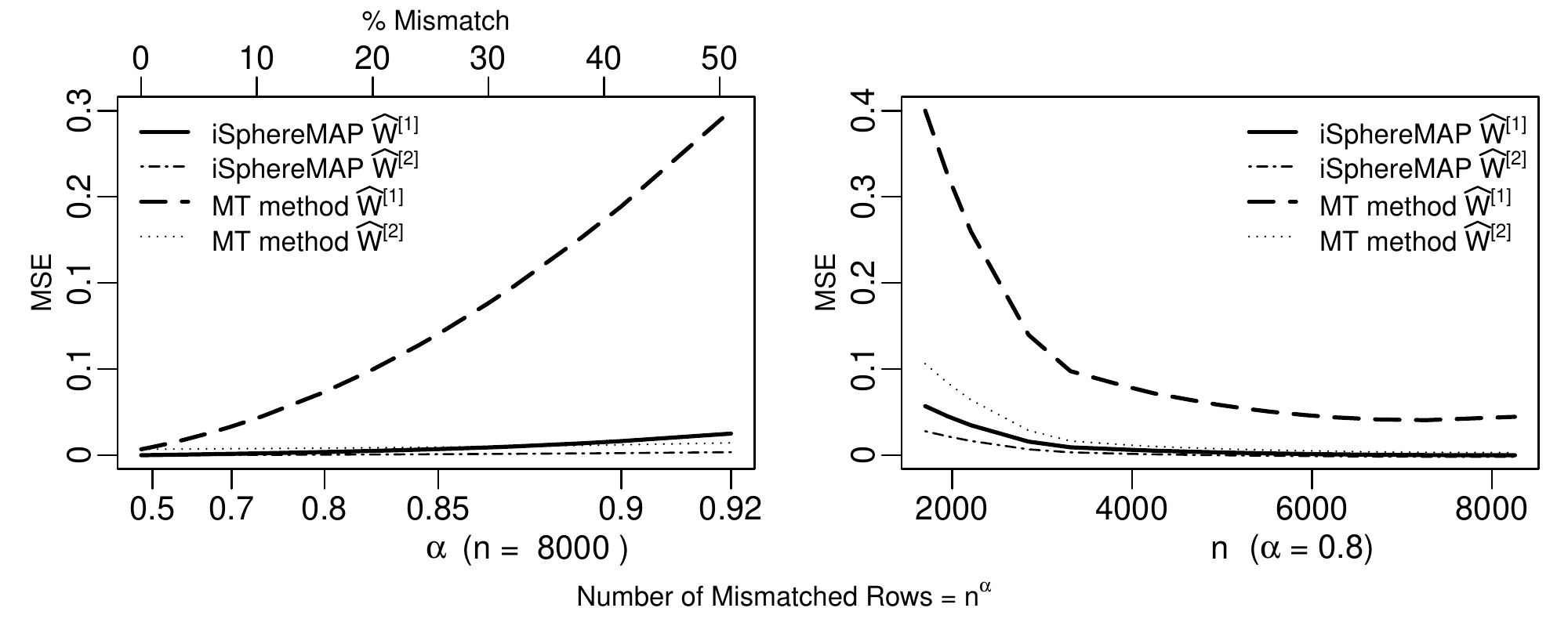}}
\caption{\label{supp:Wsimurslt} Performance of $\Wbbhat\supone$ and $\Wbbhat\suptwo$ obtained based on the proposed spherical regression and OLS in terms of the MSE (normalized by $p^{-1}=1/300$) under increasing amount of mismatch (left panel) and sample size (right panel).}
\end{subfigure}
\begin{subfigure}[t]{1\textwidth}
\makebox[\textwidth][c]{\includegraphics[width=1\textwidth, scale=.5]{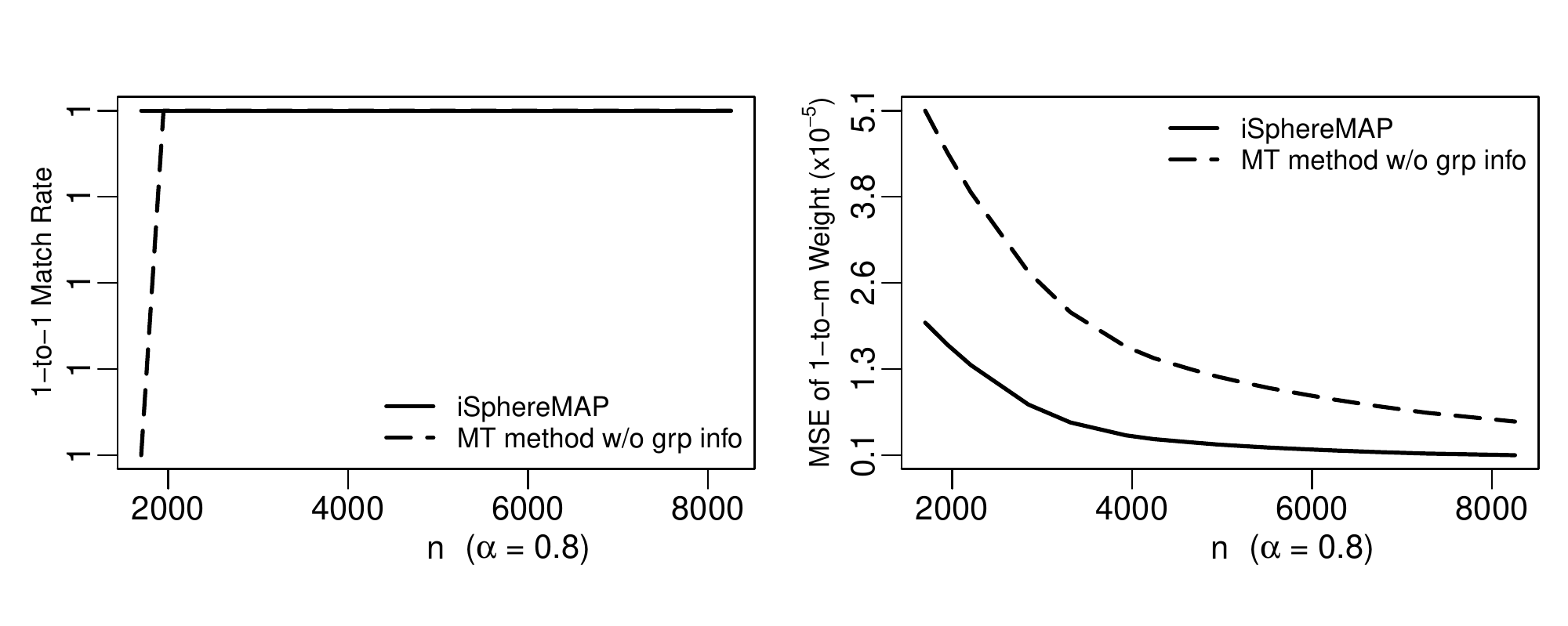}}
\caption{\label{supp:Pisimurslt} Performance of $\bPihat\suptwo$ estimated with and without group information in terms of the \onetoone match rate (left panel) and the MSE of \onetomany weight (right panel).}
\end{subfigure}
\caption{\label{supp:lessnoisefigs} Performance of $\Wbbhat\supone$, $\Wbbhat\suptwo$, and $\bPihat\suptwo$ in a less noisy scenario.}
\end{figure}

We summarize in Figure~\ref{supp:Wsimurslt} the MSEs of $\Wbbhat\supone$ and $\Wbbhat\suptwo$ from spherical regression and the MT method which uses the OLS. Figure~\ref{supp:Pisimurslt} presents the performance of $\bPihat\suptwo$ estimated from our method with group information and from the MT method without group information, in terms of the match rate for \onetoone mapping and the MSE of \onetomany mapping weight. Despite the fact that the estimators have relatively less MSE and match rate with less noise, we have the same observation as in Section~\ref{simu} that the iSphereMAP procedure generally outperforms the MT method, and the refinement of $\Wbb$ reduces the MSE.

\subsection{Refinement of $\Wbb$ using all \onetoone  (mis)matched data}\label{supp:add_mismatch_1to1}
In the refined estimation of $\Wbb$, we only use data deemed correctly matched according to $\bPihat\suptwo$ to obtain $\Wbbhat\suptwo$.
As discussed in Section~\ref{discussion}, removing mismatched codes yields negligible information loss under the current setting of sparse mismatch with $\nmis=o(n)$. However, for settings with a large amount of mismatch, it may potentially improve estimation if both one-to-one matched and mismatched data are used to obtain $\Wbbhat\suptwo$, i.e., adding $\bY_i$ indexed by $\{i\notin \mathcal{C}: \bPihat\suptwo_{i\cdot}\ne \Ibb_{i \cdot} \}$ and mapping them to the corresponding $\bX_i$ according to $\bPihat\suptwo$. 
Further including the \onetoone mismatched data may increase sample size and improve estimation.
\begin{figure}[!h]
	\begin{centering}
		\makebox[\textwidth][c]{\includegraphics[width=1\textwidth,
			scale=.5]{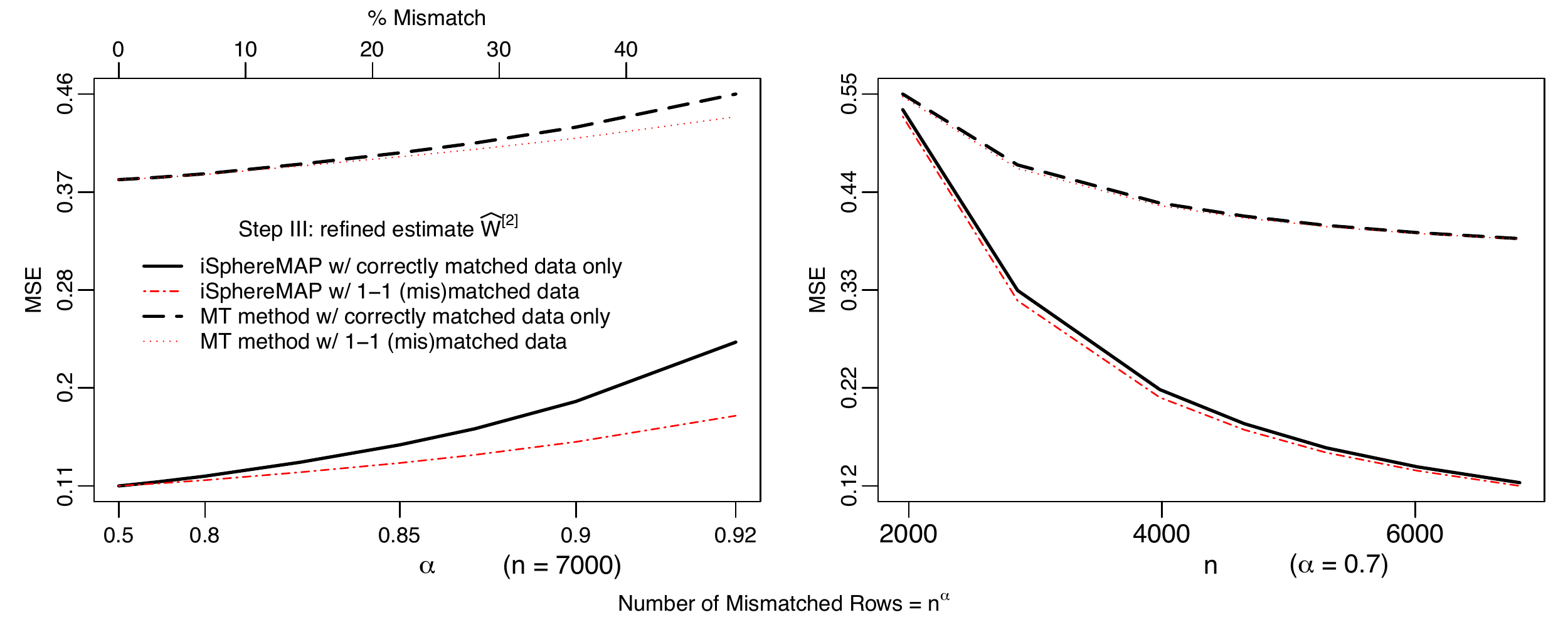}}
		\par \end{centering}
	\caption{\label{supp:fig:add_mismatch_1to1} Performance of $\Wbbhat\suptwo$ in terms of the MSE (normalized by $p^{-1}=1/300$) under increasing amount of mismatch (left panel) and sample size (right panel). Comparing $\Wbbhat\suptwo$ estimated using all \onetoone (mis)matched data versus using just correctly matched data.} 
\end{figure}

In this section, we evaluate the performance of $\Wbbhat\suptwo$ that is obtained using all \onetoone (mis)matched data, and compare it to the proposed method which uses just the \onetoone correctly matched data. 
We investigate whether further including the \onetoone mismatched data would lead to better estimation. 
Figure~\ref{supp:fig:add_mismatch_1to1} presents the MSEs of the refined estimate $\Wbbhat\suptwo$ from spherical regression and the MT method which uses the OLS. 
As is shown in Figure~\ref{supp:fig:add_mismatch_1to1}, there is some improvement for both methods in estimation of $\Wbbhat\suptwo$ when one uses all data that can be one-to-one mapped, and such improvement increases as the amount of mismatch increases. 

\newpage
\section{Alternative thresholding methods}\label{supp:adaptive_threshold}
The hard-thresholding procedure provides a framework to estimate the mapping matrix $\bPi$. Although a fixed threshold was proposed for model selection, i.e., to distinguish between \onetoone and \onetomany mappings, data-dependent adaptive thresholding  may further improve the performance.
In this section, we consider alternative strategies of thresholding that borrow information from (1) the group size $n_k$, (2) prior knowledge about the amount of \onetoone mapping within each group, and (3) the initial estimate $\Initialpi_{i\cdot}$, which we explain as follows.

As is stated in Theorem~\ref{thm:thresholding}, the hard-thresholding procedure needs to be insensitive to the estimation error of $\Initialpi$. Specifically, we require $\lambda_n \gg c_n$, where $c_n$ represents the order of $\max_{1\leq k \leq K}\|\Initialpi^k-\bPi^k\|_F$, which grows with group size $n_k$. A potential strategy to incorporate group size information is to define a group-specific threshold 
\vspace{-0.1in}\[\vspace{-0.1in}
\lambda_{n,k}\propto {\log(n_k)}\lambda_n,
\]
where $\lambda_n$ denotes an overall tuning parameter selected through cross-validation, and $\log(n_k)$ is chosen to introduce modest adjustment to the threshold based on group size. In addition, if we know a priori that group $k$ contains many \onetoone mappings, i.e., a large amount of true $\beta_i$ is zero, then it may help to use a larger threshold which encourages thresholding $\widetilde{\beta}_i$ to zero. One way to incorporate prior knowledge about the amount of \onetoone mapping is the following group-specific threshold
\vspace{-0.1in}\[\vspace{-0.1in}
\lambda_{n,k}\propto \eta_k \lambda_n,
\]
where $\eta_k$ is the proportion of \onetoone mapping in group $k$ assumed to be known a priori.

In practice, we may not have prior knowledge about the amount of  \onetoone mapping. In this case, we could consider learning the ``flatness" of the initial estimate $\Initialpi_{i\cdot}$. The ``flatness" of $\Initialpi_{i\cdot}$ indicates how distinguishable it is from a \onetoone mapping. We consider the following adaptive threshold
\vspace{-0.05in}\[\vspace{-0.01in}
\lambda_i\propto
\mathbbm{1}(\max_{j\in G_k} \Initialpi_{ij}>0)
\left\|
\frac{\Initialpi_{[i,G_k]}}{\max_{j\in G_k} \Initialpi_{ij}}	-\mathbf{1}_{n_k} 
\right\|_2\lambda_n,
\]
where $G_k$ indexes group $k$ to which item $i$ belongs, $\mathbf{1}_{n_k}$ is  a vector of ones with length $n_k=|G_k|$.
The term $\mathbbm{1}(\max_{j\in G_k}  \Initialpi_{ij}>0)$ ensures that when $\max_{j\in G_k}  \Initialpi_{ij}\leq 0$,  $\lambda_i=0$ and thus the corresponding mapping is one-to-many. The term $\|\frac{\Initialpi_{i\cdot}}{\max_j \Initialpi_{ij}}	-\mathbf{1}_{n_k} \|_2$ aims to pick up the following two patterns:
\begin{itemize}
	\item The more flat $\Initialpi_{i\cdot}$ is, the smaller $\|\frac{\Initialpi_{i\cdot}}{\max_j \Initialpi_{ij}}	-\mathbf{1}_{n_k} \|_2$ is. This will lead to a smaller $\lambda_i$ such that it is more likely to threshold $\beta_i$ to \onetomany mapping;
	\item The larger $n_k$ is, the larger $\|\frac{\Initialpi_{i\cdot}}{\max_j \Initialpi_{ij}}	-\mathbf{1}_{n_k} \|_2$ tends to be. This will lead to a larger $\lambda_i$ such that it is more likely to threshold $\beta_i$ to \onetoone mapping.
\end{itemize}

We evaluate the performance of data-adaptive thresholding method via simulation studies. Data are generated following the same procedure as Section~\ref{simu} with sample size $n=7000$ and amount of mismatch $\nmis=n^{0.88}$ for methods (1) and (3), $\nmis=n^{0.92}$ for method (2). Compared to the original iSphereMAP algorithm, the adaptive threshold will result in a different set of \onetoone and \onetomany mappings. Therefore, we evaluate the model selection performance based on the proportions of correctly identified \onetoone and \onetomany mappings. 
\begin{figure}[htp]
	\begin{subfigure}[t]{0.3\textwidth}
		\makebox[\textwidth][c]{\includegraphics[width=1\textwidth, scale=.5]{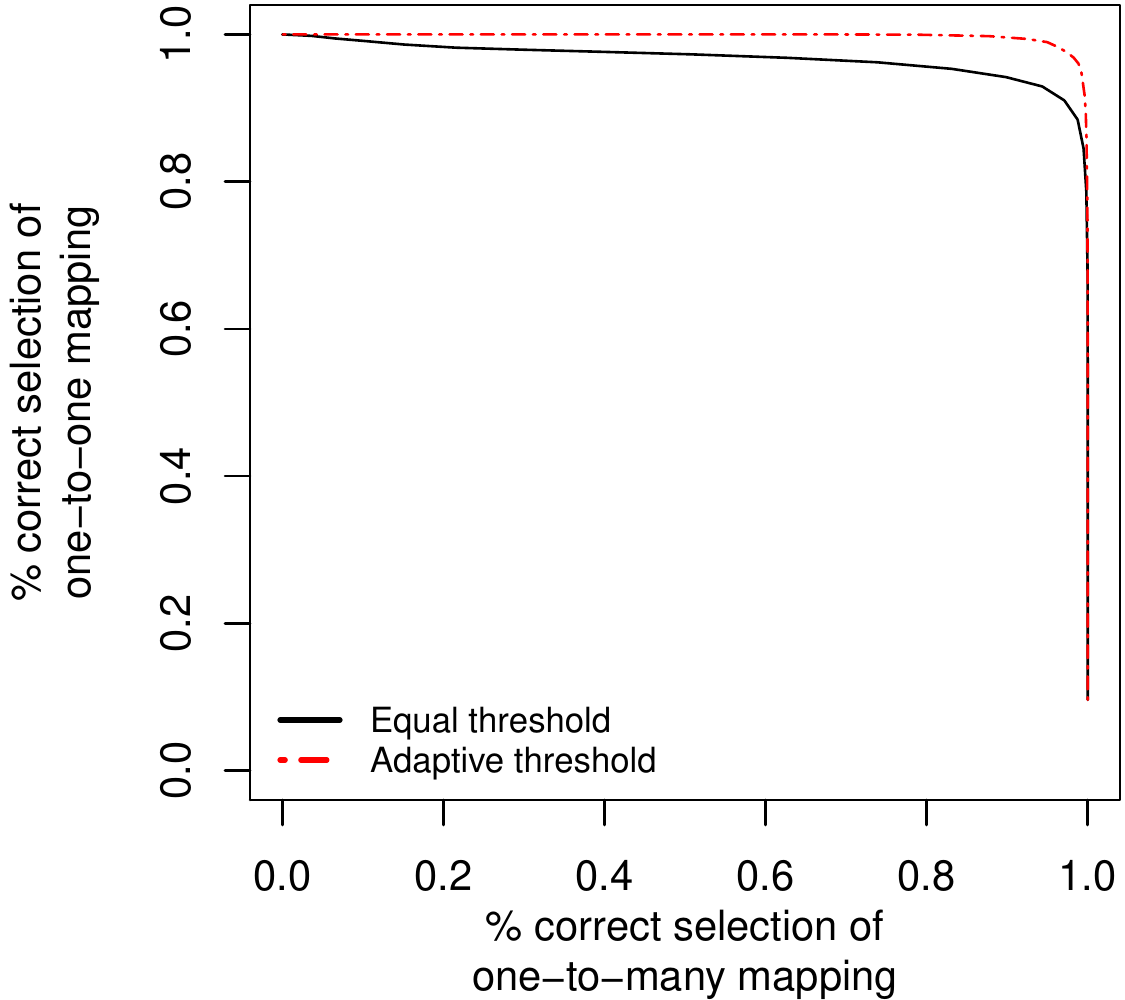}}
		\caption{\label{supp:fig:lognk} 
			$\lambda_{n,k}\propto {\log(n_k)}\lambda_n$}
	\end{subfigure}\hspace{0.2in}
	\begin{subfigure}[t]{0.3\textwidth}
		\makebox[\textwidth][c]{\includegraphics[width=1\textwidth, scale=.5]{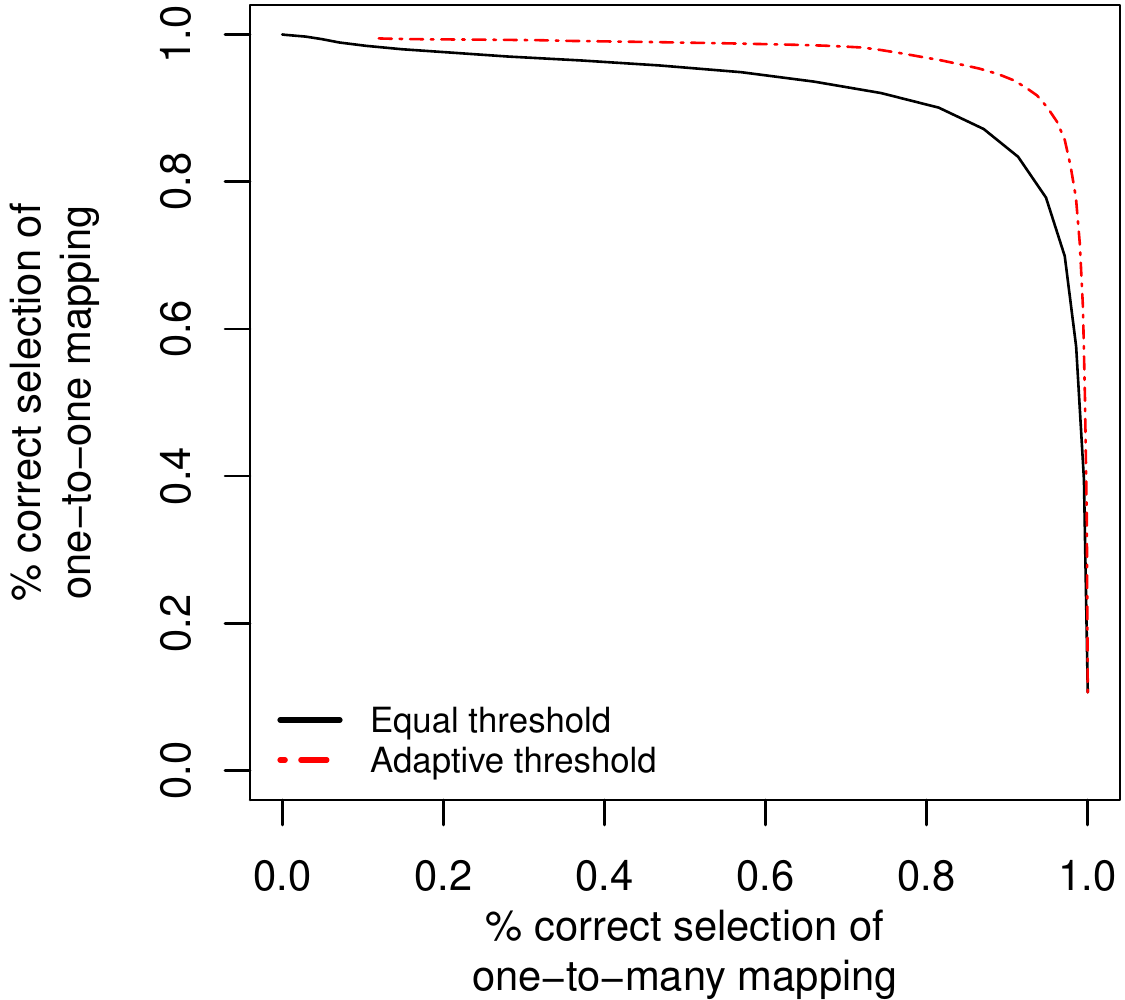}}
		\caption{\label{supp:fig:known_amount_1tom} 
			$\lambda_{n,k}\propto \eta_k \lambda_n$}
	\end{subfigure}\hspace{0.2in}
	\begin{subfigure}[t]{0.3\textwidth}
		\makebox[\textwidth][c]{\includegraphics[width=1\textwidth, scale=.5]{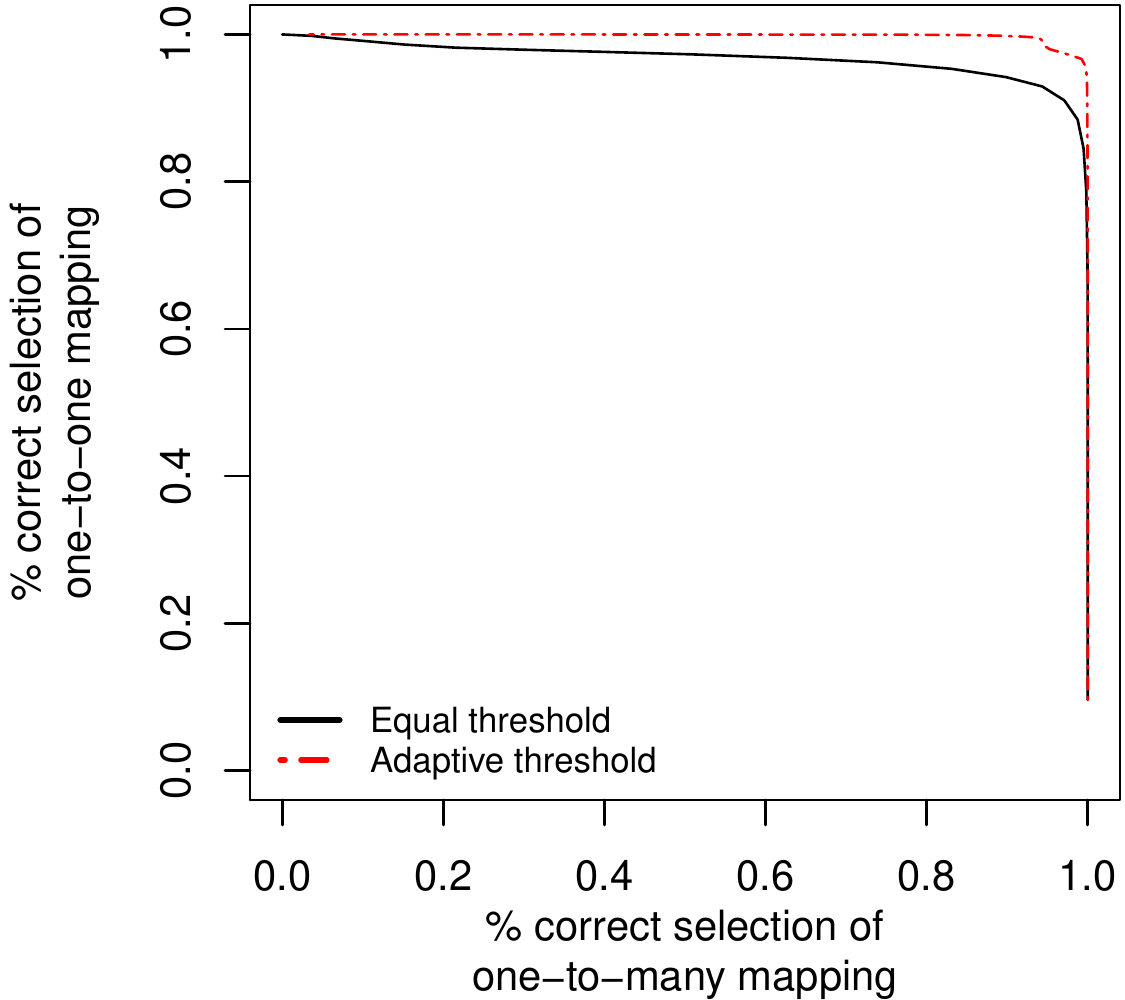}}
		\caption{\label{supp:fig:codewise_adaptive} 
			{\small{$\lambda_i\propto\mathbbm{1}\!(\max_j \!\Initialpi_{ij}>0)$\\ $\|\frac{\Initialpi_{i\cdot}}{\max_j \Initialpi_{ij}}	-\mathbf{1}_{n_k}	\|_2\lambda_n$}}}
	\end{subfigure}
	\caption{\label{all_adaptive_threshold_rslts}Performance of model selection in terms of the proportions of correctly identified \onetoone and \onetomany mappings, comparing the data-adaptive threshold and equal threshold.}
\end{figure}
Figure~\ref{all_adaptive_threshold_rslts} presents the performance of model selection for all three adaptive thresholding methods. Specifically, Figure~\ref{supp:fig:lognk} shows the performance of group-size-specific thresholding; Figure~\ref{supp:fig:known_amount_1tom} evaluates the contribution of prior knowledge on the amount of \onetoone mapping; Figure~\ref{supp:fig:codewise_adaptive} presents the performance of code-specific thresholding based on the initial estimate. All three methods have better model selection performance in terms of the percentages of correctly identified \onetoone and \onetomany mappings.

	\bibliographystyle{agsm}
	\bibliography{ref_W2V}

\end{document}